\documentclass[11pt,a4paper]{article}
\usepackage{fullpage}
\usepackage{amsmath,amssymb,graphicx,amsthm,dsfont}
\usepackage{color,soul}
\usepackage{xspace}
\usepackage{enumerate}
\usepackage{expdlist}
\usepackage{mathtools}
\usepackage{bbm}

\usepackage{booktabs} % For formal tables
\usepackage[ruled]{algorithm2e} % For algorithms

\SetAlFnt{\small}
\SetAlCapFnt{\small}
\SetAlCapNameFnt{\small}
\SetAlCapHSkip{0pt}
\IncMargin{-\parindent}

\usepackage{times}
\usepackage{soul}
\usepackage{url}
\usepackage[hidelinks]{hyperref}
\usepackage{graphicx}
\usepackage{amsmath}
\usepackage{amsthm}
\usepackage[switch]{lineno}
\usepackage{enumitem}
\usepackage{xspace}
\usepackage{breqn}

\usepackage{pifont}
\newcommand{\cmark}{\ding{51}}%
\newcommand{\xmark}{\ding{55}}%
\usepackage{amsfonts}
\usepackage{xcolor}
\usepackage{tabularx}
\usepackage{subcaption}
\usepackage{diagbox}
\usepackage{cleveref}
\usepackage{thmtools}

\usepackage{latexsym}

\newtheorem{theorem}{Theorem}
\newtheorem{proposition}[theorem]{Proposition}
\newtheorem{lemma}[theorem]{Lemma}

\newtheorem{corollary}[theorem]{Corollary}

\newtheorem{observation}[theorem]{Observation}
\newtheorem{example}{Example}
\newtheorem{remark}{Remark}

\newtheorem{conjecture}{Conjecture}

\DeclareMathOperator*{\argmax}{arg\,max}

\newcommand{\SR}[1]{}

\usepackage[numbers]{natbib}
\title{Algorithms for Stable Roommate with Externalities}
\author{Jing Leng\thanks{University of Leeds, UK. Email: jingleng112@gmail.com.} \and Sanjukta Roy\thanks{Indian Statistical Institute, Kolkata, India. Email: sanjukta@isical.ac.in}}
\date{}

\begin{document}

\maketitle

\begin{abstract}
In the roommate matching model, given a set of $2n$ agents and $n$ rooms, we find an assignment of a pair of agents to a room. Although the roommate matching problem is well studied, the study of the model when agents have preference over both rooms and roommates was recently initiated by \citet{chan2016assignment}. 
%where the valuations are additive. 
 We study two types of stable roommate assignments, namely, $4$-person stable (4PS) and $2$-person stable (2PS) in conjunction with efficiency and strategy-proofness. We design a simple serial dictatorship based algorithm for finding a 4PS assignment that is Pareto optimal and strategy-proof. However, the serial dictatorship algorithm is far from being 2PS. Next, we study top trading cycle (TTC) based algorithms. We show that variations of TTC cannot be strategy-proof or PO. Finally, as \citeauthor{chan2016assignment}~(\citeyear{chan2016assignment}) showed that deciding the existence of 2PS assignment is NP-complete,
% We complement the intractability result by designing an FPT algorithm for finding a 2PS assignment for general preferences. Finally, 
we identify preference structures where a 2PS assignment can be found in polynomial time. 
\end{abstract}

\section{Introduction}
In the roommate matching problem, 
$2n$ agents (roommates) are assigned to $n$ rooms. Each room is assigned to two agents. This problem extends the Stable Marriage problem that was introduced in the seminal work by \citet{gale1962college}. The objective is to find a stable and disjoint pairing of agents such that no two individuals would prefer each other over their current matches.
It has inspired significant research in areas such as roommate matching \cite{irving1985efficient}, 3-dimensional matching \cite{mckay2021three}, and coalition formation \cite{knuth1976mariages}. We deviate from the classical roommate matching model in two ways.
First, the classical definition of stability may not be appropriate, e.g., when we assign property rights, as observed by \citet{alcalde1994exchange}. Suppose that the matching represents an assignment of persons to two-bedroom apartments. Then two people, say Alice and Bob, who prefer each other may not be able to deviate from the assignment because they may not find a new apartment that they could share. Instead, Alice may want to exchange room with Bob's partner. Thus, we may consider a matching to be stable if no two agents would prefer to have each other’s partner rather than the partner given by the matching. That is, a pair of agents would not prefer to swap their assigned positions in the matching -- this notion is known as exchange stability in the literature \cite{Cechlarova2002,CechlarovaManlove2005,chen2019reaching}. 

Second, most studies on roommate matching prioritize stability based on agents' mutual preferences, while often disregarding preferences over specific rooms. Extensions to three-dimensional matching typically focus on forming stable coalitions of size three, yet frequently overlook settings in which some entities (e.g., rooms) do not possess preferences.
This gap is particularly relevant in online marketplaces for short- and long-term accommodations, such as Airbnb, Roomi, Roomsurf, SpareRoom, etc. \cite{Airbnb,Roomi,RoomsSurf,SpareRoom}, where users can express preferences not only for roommates, but also for specific rooms.  To address this gap, recent studies on matching roommates and rooms have explored models that assume utilities over both rooms and roommates \cite{chan2016assignment,gan2023your,hosseini2024strategyproofmatchingroommatesrooms}.
In this new preference model, the definitions of stability need to be revisited. \citet{chan2016assignment} introduced two definitions of stability that incorporate agents' preferences for both other agents and rooms, considering two versions of stability. A stronger requirement called 2-person stability (2PS) (similar to exchange stability), and a weaker one called 4-person stability (4PS) were introduced. In a 2-person stable matching, no two agents prefer the room and roommate of the other agent over their own. In contrast, 4-person stability allows deviating agents to exchange positions only if the two non-deviating roommates consent to the exchange.

 % In scenarios where the set of rooms is fixed, considering the stability notions where no pair of agents would prefer to swap their assigned positions becomes particularly relevant. \SR{Add more here}

We study the two stability notions that extend the definition of exchange stability in the new preference model, in conjunction with two fundamental concepts from economics and game theory, namely, Pareto optimality (PO) and strategy-proofness (SP). 
Pareto optimality can be defined as the state in which resources are allocated as efficiently as possible so that improving one agent without worsening other(s) is not possible. Thus, 
Pareto optimality represents a measure of economic efficiency. %As a measure of efficiency, 
Previous works~\cite{chan2016assignment,hosseini2024strategyproofmatchingroommatesrooms} have considered only social welfare, which is a strict subset of Pareto optimal solutions.
In this article, we extend our understanding of efficient roommate matching mechanisms under the requirement of being strategy-proof. Agents may sometimes misreport their preferences to influence the assignment process and secure a more favorable outcome. Mechanism design and, more broadly, game theory, aim to develop procedures where truthfully reporting preferences is the dominant strategy for agents~\cite{HurwiczReiter2006}.
A mechanism is strategy-proof if no agent can obtain a better room or roommate by misreporting their preference. 
It is known that the classical stability as defined by \citet{gale1962college} and exchange stability for the roommate matching problem are not compatible with Pareto optimality; neither is classical stability compatible with strategy-proofness. %In particular, the concept of strategy-proofness ensures that no agent can improve their outcome by misrepresenting preferences.

\subsection{Our Contribution}
 In this paper, we consider the problem of assigning roommates and rooms where the preferences of agents over
other agents and rooms are given by cardinal values and agent utilities are additive. We study two algorithms, serial dictatorship and top trading cycle for the stability notions, $2$-person stable (2PS) and $4$-person stable (4PS).

In contrast to known results, we present a simple serial dictatorship-based algorithm for finding a 4PS assignment that is Pareto optimal
and strategy-proof. We show that a solution produced by such an algorithm is far from being $2$-person stable.

We then further explore the connection of our problem to problems related to house allocation as we assume that rooms do not have preferences, akin to houses. An algorithm that has been extensively studied for its property of being Pareto optimal and strategy-proof is the Top Trading Cycle (TTC) mechanism. We show that a naive modification of TTC may not terminate in our model. 
First, we introduce a new idea of contractual top trading where, due to externality, a pair of agents and a room act as a coalition and thus, both agents in the coalition need to approve any trading. Although our trading cycle mechanism with contractual trading (CTTC) always terminates, it fails to satisfy stability, Pareto optimality, or strategy-proofness. We finally design a multi-level contractual trading cycle based algorithm (CTTCR) that produces a 4PS assignment but is not Pareto optimal or strategy-proof.
A summary of the properties of our algorithms is given in \Cref{tab:al_sum}.
There exist Pareto optimal and strategy-proof (but not stable) mechanisms for school choice~\cite{10.1257/000282803322157061}, which is based on the top trading cycle algorithm.
We leave it open whether there exists an appropriate modification of TTC that is Pareto optimal and strategy-proof for our roommate model. 
Finally, if the room or roommate preferences are given based on a set of yes/no questions, then it is more likely that an agent either likes or dislikes a room or a roommate, leading to binary valuations. We show that in symmetric, binary preferences, we can find a 2PS assignment -- which is a stronger condition than a 4PS assignment--by a simple swap based algorithm in polynomial time. The known results in the literature show the hardness of finding a 2PS assignment~\cite{chan2016assignment,chen2019reaching} even when agents are indifferent towards the rooms. Our work identifies the first simple preference structure when 2PS assignments are tractable using a simple and practical algorithm. However, it is not Pareto optimal or strategy-proof.

\begin{table}[t]
    \centering
    \begin{tabular}{|c | c c c c c|}
    \hline
         & \textbf{2PS} & \textbf{4PS} & \textbf{PO} & \textbf{SP} & \textbf{Valuation}\\
         \hline
    \textbf{SD}& \xmark  & \cmark & \cmark & \cmark & General\\
    &Th.\ref{thm:SD2PSblocking} & Lem.\ref{lem:SD4PS} & Lem.\ref{lem:SD_PE}& Lem. \ref{lem:SD4PS-SP}&\\
    \hline
    \textbf{CTTC}&\xmark&\xmark&\xmark&\xmark&General\\
    & Ex.\ref{eg:CTTCN4PS} & Ex.\ref{eg:CTTCN4PS} & Ex.\ref{eg:ttcParetoEg}& Ex.\ref{eg:CTTCnSP} &\\
    \hline
    \textbf{CTTCR }&\xmark&\cmark&\xmark&\xmark&General\\
    & Ex.\ref{eg:CTTCRn2PS} & Th.\ref{the:TTCRSW_4PS} & Ex.\ref{eg:ttcParetoEg}& Ex.\ref{eg:CTTCnSP} &\\
    \hline
        \textbf{Swapping} & \cmark & \cmark & \xmark & \xmark
    % \footnote{We can leave it as an open question that if there is any \textit{Selection Rule F} can make algorithm be strategy-proof.} 
    & Symmetric \& Binary \\
     &Th.\ref{the:SB_SW} & Th.\ref{the:SB_SW} & Ex.\ref{ex:noPObinaryswap}& Ex. \ref{ex:eg_SB_n_SP_2}&\\
    \hline
    \end{tabular}
    \caption{ Summary of Our Algorithms (Th., Lem. and Ex. refer to Theorem, Lemma, and Example, respectively.). }
    \label{tab:al_sum}
\end{table}

\subsection{Further Related Works}
% \SR{
% \begin{itemize}
%     \item Chan's paper and Hosseini et al. 
%     \item house allocation
%     \item Serial dictatorship, TTC
% \end{itemize}}
 The assignment of indivisible resources among agents is a core issue in computational social choice, a field that bridges economics \cite{Thomson2011} and computer science \cite{Manlove2013}. 
Our research is related to several previous works but with some key differences. First, 
our problem model is to assign an agent pair rather than an individual to each resource, and an agent cares about both the resource and the mate it gets. The original roommate matching problem was proposed by \citet{gale1962college}, but the problem only considers the agents' preferences over their roommates. 
\citet{chan2016assignment} introduced another roommate market model considering the valuations of both room and roommate, and our research focuses on this model without considering room rent.
Second, previous work for the similar model focused on maximizing the social welfare or considered various room capacities \cite{huzhang2017online, li2019room}, but our work examines the stability and the strategy-proofness. 
In resource allocation, it is desired that the decision maker accounts for both the agents' preferences for the resources and their strategic actions.
\citet{hosseini2024strategyproofmatchingroommatesrooms} study strategy-proof algorithms for approximations of maximum social welfare under additive and Leontief utilities. Although they show a strategy-proof mechanism when utilities are binary and Leontief, our result is different as we study stability under additive utilities with binary and symmetric valuations and  general valuations. Moreover, as opposed to social welfare, we consider Pareto optimality (PO) as a measure of efficiency.

Our model can be observed as an extension of the house allocation problem where agents have preferences over houses and each agent is assigned to a house. Whereas, in our problem a pair of agents is assigned to a room. PO can be achieved in simpler settings like house allocation by the serial dictatorship (SD) procedure, also called picking-sequence~\cite{BouveretLang2014}. In this procedure, agents are arranged in a specific order, and each agent, one by one, selects their most preferred resource from the pool of remaining resources.  We extend this process and show that SD is powerful enough to even find a 4PS solution. However, the SD-based algorithm loses the PO property when there are ties.
This is undesirable because agents may have an incentive to avoid participating in the picking process. From the house allocation problems, the well-known Gale’s top trading cycle (TTC) procedure \cite{shapley1974cores} computes an allocation that is Pareto optimal.
Thus, we extend the top trading cycle algorithm to our model and examine their properties. 
 Both serial dictatorship and TTC are strategy-proof for the house allocation problem. Furthermore, \citet{Ma1994} showed that the TTC procedure is uniquely capable of being individually rational, Pareto efficient, and strategy-proof simultaneously. We show that it does not hold for multiple extensions of TTC in our model.
Other papers have studied the properties of SD and TTC \cite{morrill2024top,shapley1974cores,roth1982incentive,roth1977weak} when they are applied in a one-to-one assignment problem (e.g. house allocation). 
\citet{Fujita2018} considered a procedure similar to TTC, designed to select a Pareto optimal allocation within the core when preferences over sets of resources are restricted. When agents are allowed to exchange multiple resources, determining an allocation that is both individually rational and Pareto efficient is NP-hard for additive preferences \cite{Aziz2016}. 
 SD and TTC have been extended to an online setting, where participants enter and exit at different times \cite{lesca2025online}.
%and it is similar to real-world online marketplaces (e.g., Airbnb, online property exchanges)
Sufficient conditions for strategy-proofness in dichotomous preferences have been studied \cite{Aziz2020}.
Our study further adjusts these two algorithms and examines whether their properties still hold in our roommate matching model. 

% Gan et al.[\citeyear{gan2023your}] showed that it is NP-complete to decide the existence of an envy-free assignment.

In contrast to our results for 2PS, finding a 2PS assignment is NP-complete even when each agent is indifferent between the room~\cite{chan2016assignment}. The reduction in \citet{chan2016assignment} and the results from \cite{chen2019reaching} imply that even the question of finding a ``path to stability'' using the number of divorces (swaps of agents) as the parameter is W[1]-hard and it is NP-complete even when each agent values only a bounded number of other agents and is indifferent between the rooms.

\section{Model}
\label{sec:model}
Extending the classical roommate matching model, \citet{chan2016assignment} introduces another roommate market model, in which $2n$ agents  $N = \{a_1, a_2, \dots, a_{2n}\}$  need to be assigned to $n$ rooms $R = \{r_1, r_2, \dots, r_n \}$, with two people in each room. 
% \noindent\textbf{Notation}. We define the model as follows: 
% Let $N = \{a_1, a_2, \dots, a_{2n}\}$ be the set of agents and $R = \{r_1, r_2, \dots, r_n \}$ be the set of rooms. 
An instance $\langle N, R, H, V \rangle$ of roommate matching includes a pair of matrices $H$ and $V$ with
$H=\{h_i | h_i: N\setminus\{i\} \mapsto \mathbb{R}_{\geq0}\}$,
% $H = \{h_{ij} \; | \; i, j \in N, i \neq j\}$, 
where $h_{i}(j)$ is the happiness of person $i$ when living with $j$ as its roommate, and 
$V=\{v_i | v_i: R \mapsto \mathbb{R}_{\geq0}\}$,
% $V = \{v_{ir} \; | \; i \in N, r \in R\}$, 
where $v_{i}(r)$ is the valuation of person $i$ for room $r$.
We say a pair of agents $i,j \in N$ and a room $r$ is a triple $(i,j,r)$.
A roommate \emph{assignment} $ \mu $ is a $n$-sized set of disjoint triples, specifying an allocation of $n$ disjoint pair of agents to the $n$ rooms.  For a triple $(i,j,r)$ in assignment $\mu$, we use $\mu_N(i)$ and $\mu_R(i)$ to denote agent $j$ and room $r$ respectively. We use $U_\mu(i)$ to denote agent $i$'s utility $h_i(j) + v_i(r)$ in assignment $\mu$. 

We define the {\em swap operation} as follows.
% \noindent\textbf{Notation.} We start by defining a {\em swap operation}. 
Let $\mu$ be a roommate assignment %. Given a pair of agents $a,c \in N$, we assume that
 and let the triples $t_1=(a, b, r_1)$ and $t_2=(c, d, r_2)$ be in the assignment $\mu$ where $a,b,c,d \in N$ and $r_1,r_2 \in R$. We say an assignment $\mu'$ is created from $\mu$
 by \emph{swapping agents $a$ and $c$} if $\mu'$ is defined by deleting the triples  $t_1$ and $t_2$ from $\mu$ and adding the triples $(c, b, r_1)$ and $ (a, d, r_2)$. That is,  $\mu' = \mu \cup \{(c, b, r_1), (a, d, r_2)\} \setminus \{t_1,t_2\} $. We use $\mu_{a\leftrightarrow c}$ to denote $\mu'$, and thus, $U_{\mu(a\leftrightarrow c)}(a)$ denotes agent $a$'s utility in assignment $\mu'$, i.e. $U_{\mu(a\leftrightarrow c)}(a)= h_a(d) + v_a(r_2)$.

\noindent\textbf{Preference Structures}. There are 3 categories of preferences discussed in this paper - symmetric, binary, and general valuations. 
% By saying "preferences", we refer to the cardinal preference which assigns numerical values to the different levels of satisfaction and assumes that preference can be measured and compared precisely.
\emph{Binary} indicates that preference values are either $0$ or $1$, so for any agent $i, j \in N$ and any room $r \in R$, $h_i(j) \in \{0, 1\}, v_i(r) \in \{0, 1\}$. \emph{Symmetry} indicates that agents like each other mutually and equally, i.e., any two agents $i, j \in N$, it holds that $h_i(j) = h_j(i)$. %in symmetric valuations. 

\noindent \textbf{Stability}. 
% Our model takes agents' valuations for both other agents and rooms into account and considers two versions of stability:  a weaker stable property called  2-person stable (2PS) and a stronger one called 4-person stable (4PS). 
% When agents express preference only over the roommates, exchange stable matching \footnote{In the original Stable Roommates Problem \cite{irving1985efficient}, the notation of stability is different from ours.} refers to a situation in which no two agents both prefer each other to their current roommates. 
% Referring to \citet{chan2016assignment}'s work, we use the terms \emph{2PS blocking pair and 4PS blocking pair} to refer to the two agents breaking the stability conditions, respectively.
% 
% 
% We provide formal definitions of 2PS and 4PS. 
%
 % \noindent \textbf{2-person stability}. 
 \citet{chan2016assignment} extended the concept of exchange stability\footnote{When agents' valuations are over the roommates and they are indifferent about the rooms, exchange stable matching refers to an assignment in which no two agents both prefer the others agent's roommate to their current roommate.} to include the preferences over the rooms. An assignment is 2-person stable, or 2PS for short, if no pair $(i, j)$ of agents living in different rooms can swap and increase both their utilities. That is, an assignment $\mu$ is 2PS if for any two agents $i$ and $j$, we have that $ h_i(j') + v_i(s) \leq h_i(i') + v_i(k)$ or $h_j(i') + v_j(k) \leq h_j(j') + v_j(s)$ where the triples $(i,i',k)$ and $(j,j',s)$ are in $\mu$.
 % we assume that agent $i$ and $i'$ are assigned to room $k$, and agent $j$ and $j'$ are assigned to room $s \neq k$, then

 % \noindent\textbf{4-person stability}. 
 The definition of 4PS is an extension of the definition of 2PS by taking the roommates of $i$ and $j$ into consideration. %, and their objection also make the swap of $i$ and $j$ infeasible. That is, the deviating agents are allowed to exchange positions only if the two non-deviating roommates allow the exchange.  
% \footnote{The definition of 4-person stability in our study is the same as the one in \citet{chan2016assignment}'s work}, 
 An assignment is 4-person stable, or 4PS for short, if no pair $(i, j)$ of agents living in different rooms can swap and make all 4 people in the two rooms increase their utilities. That is, an assignment $\mu$ is 4PS if for any two agents $i$ and $j$, we have that  $h_i(j') + v_i(s) \leq h_i(i') + v_i(k)$ or $h_j(i') + v_j(k) \leq h_j(j') + v_j(s)$ or $h_{i'}(j) \leq h_{i'}(i)$ or $h_{j'}(i) \leq h_{j'}(j)$ where the triples $(i,i',k)$ and $(j,j',s)$ are in $\mu$.
 % weakly increase their utilities and both $i$ and $j$ strictly improves. That is, an assignment $\mu$ is 4PS if for any two agents $i$ and $j$, we have that  $h_i(j') + v_i(s) \leq h_i(i') + v_i(k)$ or $h_j(i') + v_j(k) \leq h_j(j') + v_j(s)$ or $h_{i'}(j) < h_{i'}(i)$ or $h_{j'}(i) < h_{j'}(j)$ where the triples $(i,i',k)$ and $(j,j',s)$ are in $\mu$.
% 
 % if agents $i$ and $i'$ are assigned to room $k$, and agents $j$ and $j'$ are assigned to room $s \neq k$, then $h_i(j') + v_i(s) \leq h_i(i') + v_i(k)$ or $h_j(i') + v_j(k) \leq h_j(j') + v_j(s)$ or $h_{i'}(j) \leq h_{i'}(i)$ or $h_{j'}(i) \leq h_{j'}(j)$.
%
 % Therefore, a 4-person stable assignment refers to an assignment where for any four agents, there is no potential swap that can weakly benefit all four involved agents.
 We use the terms \emph{2PS blocking pair and 4PS blocking pair} to refer to the two agents breaking the stability conditions in an unstable assignment, respectively.

%Chan's model is the model we mainly study on. 
\citet{chan2016assignment} showed that it is NP-complete to decide the existence of a 2PS assignment and a solution satisfying the 4PS property always exists and can be found in time a $O(n^2)$.

\noindent\textbf{Social Welfare}. In the context of the roommate matching problem, \emph{Social Welfare} (SW)  is the sum of the utility of all participants in the matching. %In our model, each agent's utility is obtained from its valuations to its roommate and room. 
% A given assignment $\mu$ is a set of $n$ triples $(i, j, r)$, where $i, j \in N$ and $r \in R$, and
Thus, the social welfare of assignment $\mu$ is defined as:

%\SR{First say what is the social welfare of a given matching $\mu$} \JL{Done}

$ SW(\mu) = \sum_{i\in N :\\(i,j,r) \in \mu}  h_i(j)+ v_i(r)$.

% \[\text{Alternatively, }  SW(\mu) = \sum_{(i,j,r) \in \mu}  h_i(j) + v_i(r) + h_j(i) + v_j(r)\]

% With agents $i, j \in I$ and room $r \in R$, let $x_{ij}$ be a binary decision variable that indicates whether individuals $i$ and $j$ are paired together. We define $x_{ij} = 1$ if agent $i$ and agent $j$ are paired together, otherwise $x_{ij} = 0$. ($x_{ij} \in \{0, 1\}$ for all $i, j$). Similarly, we use $y_{ir}$ to represent a binary decision variable that indicates whether agent $i$ is allocated to room $r$ ($y_{ir} \in \{0, 1\}$). The equation below can be used to calculate the social welfare: 
% \[ SW = \sum_{i=1}^{2n} \sum_{j=1}^{2n} \sum_{r=1}^{n} (h_{ij} + v_{ir} + h_{ji} + v_{jr}) \cdot x_{ij} \cdot y_{ir} \]

\noindent\textbf{Pareto Optimality}. Pareto Optimality (PO) is a concept from economics and describes the efficiency of resource allocation \cite{stiglitz1981pareto}. 
% It aligns with economic principles of maximizing resource allocation without making anyone worse off. 
An assignment $\mu'$ Pareto dominates assignment $\mu$ if at least one agent's utility strictly increases and utility of no agent decreases in $\mu'$ as compared to $\mu$.
An assignment is Pareto Optimal (or Pareto Efficient) if no assignment Pareto dominates it. 

\noindent\textbf{Strategy-proofness}.
Strategy-proofness, also known as incentive compatibility, is a desirable property in matching mechanisms. %where participants have an incentive to truthfully reveal their preferences. In the context of the roommate matching problem, 
A strategy-proof mechanism ensures that agents cannot benefit by misreporting their valuations.
%
% Given an instance $\langle N, R, H, V \rangle$ of roommate matching 
We use $h_{-i}$ and $v_{-i}$ to denote the valuations of all agents in $N \setminus \{i\}$. Let $\mathcal{A}$ be a mechanism that computes a roommate assignment given an input instance $I = \langle N, R, H, V \rangle$. Recall, $U_{\mathcal{A}(I)}(i)$ denote the utility of agent $i$ in assignment $\mathcal{A}(I)$. A mechanism $\mathcal{A}$ is \emph{strategy-proof} if, for any agent $i$, for each compatibility value  $h_i'$ and room value  $v_i'$,  it holds that $U_{\mathcal{A}(I)}(i) \geq U_{\mathcal{A}(I')}(i)$ where $I' = \langle N,R, h_{-i}\cup \{h_i'\}, v_{-i} \cup \{v_i'\} \rangle$.

Therefore, no agent can improve their utility by misreporting their valuations. An algorithm that realizes a strategy-proof mechanism is called a strategy-proof algorithm. 
%
%\SR{Add other definitions like social welfare.
%Preferences: 
%DONE}
% \section{Algorithmic Developments}
%\JL{Q-Should I mention the content in 3.2-3.4 here? Done}\SR{Yes}
\subsection{Structural Properties}
We first explore the relation between PO and the stability notions.

\begin{proposition}\label{POimply4PS}
    Every Pareto optimal assignment is a 4PS assignment but not all 4PS assignments are Pareto optimal.
\end{proposition}

\begin{proof}
    Suppose that an assignment $\mu$ is PO but not 4PS, and there is a 4PS blocking pair $(i, j)$. Let the triples containing the agents $i$ and $j$ in assignment $\mu$ be $t_1 = (i, i', r_1)$ and $t_2 = (j, j', r_2)$. According to the definition of 4PS, the swap of agent $i$ and $j$ increases the utilities of all four agents $i, i', j$, and $j'$. Moreover, the utilities of all other agents remain the same in assignment $\mu_{i\leftrightarrow j}$. 
    Then,  the assignment $\mu_{i\leftrightarrow j}$ Pareto dominates the assignment $\mu$, contradicting the Pareto optimality of $\mu$. 
    To show the second part of the statement, we show Example \ref{eg:4PSnPO} where all agents in a 4PS assignment can be improved. 
\end{proof}
The above proposition also implies that an assignment with maximum social welfare is 4PS. However, the converse is not true.

To analyze how far a 4PS assignment can be from Pareto optimality, we measure the number of agents who can improve without worsening others. Example \ref{eg:4PSnPO} shows that in the worst-case scenario, all agents in a 4PS assignment can be improved. Additionally, the social welfare of a certain 4PS assignment can be an arbitrarily small fraction of the maximum social welfare.

% \begin{table}[t]
%     \centering
%     \begin{tabular}{|m{3mm}|m{3mm}|m{3mm}|m{3mm}|m{3mm}|m{3mm}|m{3mm}|}
%     \hline
%          & $a$ & $b$ & $c$ & $d$ & $r_1$ & $r_2$\\
%          \hline
%         $a$ &--{}{} & 1 & 0 & 0 &0 & 1 \\
%         \hline
%         $b$  &0 &--{}{}& $\infty$ & 0 & 0 & 1\\
%        \hline
%         $c$ & 0 & 0 &--{}{}& 1& 1 & 0 \\
%         \hline
%        $d$ & 0 & 0 & 1 &--{}{}& 1 & 0  \\
%        \hline
%     \end{tabular}
%     \caption{Valuations: an example where a 4PS assignment is not PO}
% \label{tab:4PSnPO}
% \end{table}

\begin{example}
\label{eg:4PSnPO}
    Consider an instance with four agents $a, b, c$ and $d$, two rooms $r_1$ and $r_2$, and the valuations given in \cref{tab:4PSnPO}. Consider assignment $\mu_1 = \{(a, b, r_1), (c, d, r_2)\}$. It is 4PS. However, an assignment $\mu_2 = \{(c, d, r_1), (a, b, r_2)\}$ can benefit all the agents simply by swapping their rooms. Thus, assignment $\mu_1$ is not PO.
    
    Furthermore, we consider an extreme case where agent $b$ has a large preference value for agent $c$. e.g., say, $h_b(c) = L$ for some $L \rightarrow \infty$. A 4PS assignment $\mu_3 = \{(b, c, r_1), (a, d, r_2)\}$ gives a maximum social welfare close to $\infty$, compared to arbitrarily small social welfare of $\mu_1$. 
\end{example}
% While any PO or maximum social welfare assignment must necessarily be 4PS, the converse does not hold. Moreover, the gap between 4PS and these optimality criteria can be arbitrarily large, and a 4PS assignment may fall far short of achieving PO or efficiency.

Furthermore, \Cref{eg:n2PS_SP} shows that the stricter stability notation 2PS is incompatible with strategy-proofness. 

\begin{proposition}
2PS is incompatible with strategy-proofness. 
\end{proposition}

\begin{table}[t]
    \centering
    \begin{minipage}{0.48\textwidth}
        \centering
        \begin{tabular}{|m{3mm}|m{3mm}|m{3mm}|m{3mm}|m{3mm}|m{3mm}|m{3mm}|}
        \hline
             & $a$ & $b$ & $c$ & $d$ & $r_1$ & $r_2$\\
             \hline
            $a$ &--{}{} & 1 & 0 & 0 &0 & 1 \\
            \hline
            $b$  &0 &--{}{}& $\infty$ & 0 & 0 & 1\\
           \hline
            $c$ & 0 & 0 &--{}{}& 1& 1 & 0 \\
            \hline
           $d$ & 0 & 0 & 1 &--{}{}& 1 & 0  \\
           \hline
        \end{tabular}
        \caption{Valuations: an example where a 4PS assignment is not PO}
    \label{tab:4PSnPO}
    \end{minipage}
    \hfill
    \begin{minipage}{0.48\textwidth}
        \centering
        \begin{tabular}{|m{3mm}|m{3mm}|m{3mm}|m{3mm}|m{3mm}|m{3mm}|m{3mm}|}
            \hline
                 & $a$ & $b$ & $c$ & $d$ \\
                 \hline
                $a$ &--{}{} & 3 & 2 & 1 \\
                \hline
                $b$  &2 &--{}{}& 3 & 1 \\
               \hline
                $c$ & 3 & 2 &--{}{}& 1\\
                \hline
               $d$ & 2 & 1 & 3 &--{}{}\\
               \hline
            \end{tabular}
            \caption{Valuations: an example demonstrating the incompatibility of 2PS and strategy-proofness}
        \label{tab:n2PS_SP}
    \end{minipage}
\end{table}

\begin{example}
    \label{eg:n2PS_SP}
    Consider an instance with four agents $a, b, c$ and $d$, and the valuations given in \cref{tab:n2PS_SP}. We suppose that the agents are indifferent between rooms and omit the valuations for rooms. Among all possible assignments, there are two 2PS assignments $\mu_1 = \{(a, b), (c, d)\}$ and $\mu_2 = \{(a, d), (b, c)\}$. However, in assignment $\mu_1$, agent $b$ can gain by lying (e.g. making $v_b(a)=1$ and $v_b(d)=2$), so that the algorithm will output the only 2PS assignment $\mu_2$. Similarly, in assignment $\mu_2$, agent $d$ can gain by lying so that $\mu_1$ is the only 2PS assignment. 
    \end{example}

\noindent
\textbf{Organization of the paper.} In \cref{sec:DS_LS}, we first examine the algorithms in \cite{chan2016assignment} for maximizing social welfare and for 4PS. We provide examples to show that they are not strategy-proof. In \cref{sec:SD} and \ref{sec:TTC}, we modify the classical Serial Dictatorship algorithm and Top Trading Cycle algorithm to fit our roommate matching model. Additionally, we examine the properties of the newly developed algorithms, including stability and efficiency. Finally, in \cref{sec:sym_bin}, we consider a specific structure where valuations are symmetric and binary, and design an algorithm to produce a 2PS assignment.   

\section{An algorithm for 4PS Assignment}
\label{sec:DS_LS}
\citet{chan2016assignment} developed an algorithm to produce a 4PS assignment that has social welfare at  least $2/3$ of the optimal. It uses first an algorithm called Double Matching (DM) and then the Local Searching (LS) algorithm running in time $O(n^3)$  and $O(n^5)$, respectively. 
Although Chan et al. \cite{chan2016assignment} designed a greedy algorithm following the serial dictatorship mechanism to find a 4PS solution efficiently, we provide a counterexample to prove that it is not strategy-proof.
\begin{restatable}{proposition}{propDM}
\label{prop:DM}
    Double Matching and Local Searching algorithms in \cite{chan2016assignment} are not strategy-proof.
\end{restatable}
\Cref{eg:doublematching} in the appendix proves the above proposition.

\section{Serial Dictatorship (SD)}
\label{sec:SD}
The lack of the strategy-proofness of the previous algorithm inspires us to study Serial Dictatorship (SD) algorithm in our model. It has been extensively used to design strategy-proof mechanisms for various allocation problems, including house allocation \cite{abdulkadirouglu1998random}, object allocation \cite{svensson1999strategy}, and facility location \cite{caragiannis2024truthful}.
% Serial dictatorship (SD) algorithm is typically used in the one-to-one matching problem (e.g. House Allocation), where one good is assigned to one agent, and agents have preferences over goods.
It is well known that in most cases with two-sided preferences, stability is not compatible with 
 strategy-proofness \cite{r82}. In contrast, we show that a strategy-proof 4PS mechanism can be designed using a modification of the serial dictatorship algorithm.

\subsection{SD Algorithm}
We extend the classical SD algorithm for our model. For a given instance $\langle N,R,H,V\rangle$, our algorithm works as follows:
\begin{enumerate}
    \item Assignment Order: Fix an order of the agents, denoted as a sequence $\sigma = (a_1, a_2, \dots, a_{2n})$. %This order could be determined based on some priority criterion.
    \item Sequential Choice: The first agent in the order $\sigma$ selects its most preferred agent and room. The paired agents are allocated to the selected room, and they are removed from the market. The process continues with the next agent in the market according to the order $\sigma$ until all agents and rooms have been allocated and removed. 
    % If the second agent ($a_2$) is still in the market, it selects its most preferred agent and room from the remaining, or the process goes to the third agent if the second one was removed from the market. 
    % This process continues sequentially until all agents have been allocated and removed.
\end{enumerate}
The algorithm is also presented in \Cref{alg:SD4PS} in the appendix where we show in \Cref{lem:SD_PE,lem:SD4PS-SP} that the algorithm produces a Pareto optimal assignment for strict utilities and is strategy-proof.
Apart from being strategy-proof and PO, serial dictatorship produces a 4-person stable solution. We give the theorem and the proof below. 
% 
% \paragraph{4-Person Stability}
% We finally show that \Cref{alg:SD4PS} produces a 4PS assignment. 

\begin{lemma}\label{lem:SD4PS}
     The serial dictatorship algorithm outputs is $4$-person stable assignment.
\end{lemma}
\begin{sloppypar}
\begin{proof}
    We prove this by contradiction. Let $\mu$ be an assignment output by \Cref{alg:SD4PS}. Suppose that there is a 4-person stable blocking pair $(i, j)$, and the triples involved in the assignment are $(i, i', r)$ and $(j, j', s)$. According to the definition of 4PS blocking pair, the swap of agent $i$ and agent $j$ can increase the utilities for all four agents $i, i', j$, and $j'$, i.e., $h_i(j') + v_i(s) > h_i(i') + v_i(r)$, $h_j(i') + v_j(r) > h_j(j') + v_j(s)$, $h_{i'}(j) > h_{i'}(i)$, and $h_{j'}(i) > h_{j'}(j)$. We show that the agent who precedes the others in $\sigma$ cannot strictly improve.
    
    We assume that $i$ precedes $i'$, $i$ precedes $j$, and $i$ precedes $j'$ in order $\sigma$, i.e., agent $i$ has the highest priority in selecting among these $4$ agents. In the assignment process, as agent $i$ chooses agent $i'$ over $j$, and chooses room $r$ over $s$, we can infer that $h_{i}(i') \geq h_i(j')$, and $v_{i}(r) \geq v_{i}(s)$. Thus, $h_{i}(i') + v_{i}(r) \geq h_{i}(j') + v_{i}(s)$, which contradicts the existence of the 4PS blocking pair. The cases where agents $i',j$, or $j'$ come first in the ordering $\sigma$ are analogous. Thus, the assignment produced by the algorithm is 4-person stable.
\end{proof}\end{sloppypar}

Thus, consolidating the lemmas, we get the following result.

\begin{theorem}
The serial dictatorship algorithm is strategy-proof, $4$-person stable,  Pareto optimal when the utilities are strict, and runs in time $O(n^2)$.
     % There exists a strategy-proof mechanism that runs in time $O(n^2)$ and produces a $4$-person stable assignment and it is Pareto optimal when the utilities are strict.
\end{theorem}

\paragraph{2-Person Stability}
Serial dictatorship does not guarantee a 2-person stable assignment. 
%
% A simple example is that we assume agent $a_1$ chooses $a_3$ as his roommate and room $r_1$ as his room, and agent $a_2$ chooses $a_4$ as his roommate and room $r_2$ as his room. Triple set $\{(a_1, a_3, r_1)$, $(a_2, a_4, r_2)\}$ denotes the assignment output by SD. However, this allocating process ignores the references of agent $a_3$ and $a_4$. It can happen that agent $a_3$ prefer agent $a_2$ (or room $r_2$) over agent $a_1$ (or room $r_1$), and similarly, agent $a_4$ prefer agent $a_1$ (or room $r_1$) over agent $a_2$ (or room $r_2$). In this case, $(a_3, a_4)$ is a 2PS blocking pair, and agent $a_3$ and $a_4$ are happy to swap with each other.
%
% Furthermore, w
%
We evaluate the worst-case scenario when applying SD with regard to 2-person stability. We give the following proof, and  \Cref{eg:SD2PSblocking} in the appendix presents a tight example. 

\begin{theorem}\label{thm:SD2PSblocking}
    % Assigning $2n$ agents to $n$ rooms by applying the serial dictatorship algorithm 
    The serial dictatorship algorithm (\Cref{alg:SD4PS}) produces an assignment with at most  $(n^2 - n)$   $2$-person blocking pairs. 
\end{theorem}

\begin{proof}
Let the assignment output by the \Cref{alg:SD4PS} be $\{(a_1, a_2, r_1)$, $(a_3, a_4, r_2)$, \dots, $(a_{2n-1}, a_{2n}, r_n)\}$. Let agent $a_{2j-1}$ denotes the dominant agent in the $j$th triple (i.e., agent $a_{2j-1}$ selects his roommate and room to form the $j$th triple) for $j \in [n]$. For agent $a_{2j-1}$ and his roommate $a_{2j}$, the following holds:
\begin{itemize}\looseness-1
   \item For any agent $a_{2i-1}$ with a higher priority than agent $a_{2j-1}$ (i.e. $i $ precedes $j$ in $\sigma$), pairs $(a_{2i-1}, a_{2j-1})$ and $(a_{2i-1}, a_{2j})$ cannot be blocking.
   \item For any agent $a_{2k-1}$ with a lower priority than agent $a_{2j-1}$ (i.e. $j$ precedes $k$ in $\sigma$), pairs $(a_{2j-1}, a_{2k-1})$ and $(a_{2j-1}, a_{2k})$ cannot be blocking. 
\end{itemize}
\begin{figure}[t]
        \centering \includegraphics[width=0.3\textwidth]{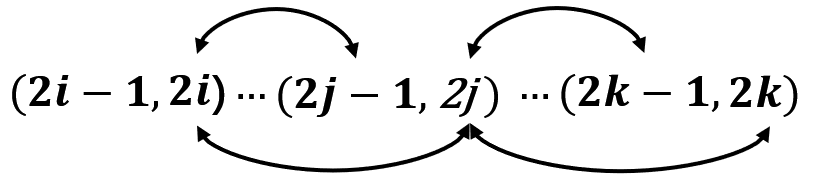}
    \caption{Possible Blocking Pairs in \Cref{thm:SD2PSblocking}}
    \label{fig:blockingpair}
        \vspace{5mm}
\end{figure}%\hfill
  
In the worst-case scenario, there exists blocking pairs $(a_{2i}, a_{2j-1})$, $(a_{2i}, a_{2j})$, $(a_{2j}, a_{2k-1})$ and $(a_{2j}, a_{2k})$, as shown in figure \ref{fig:blockingpair}. 
Thus, for triple $(a_{2j-1}, a_{2j}, r_j)$, there are $j-1$ triples before and $n-j$ triples after, making it $n-1$ pairs in total. As a result, there are $2 \times (n-1) = 2n-2$ blocking pairs involving agent $a_{2j-1}$ or $a_{2j}$. For all the agents, the number of blocking pairs is:
 $\sum_{j=1}^{n}2(n-j) = n^2 - n$.

Therefore, by applying the serial dictatorship algorithm, the output assignment has at most $(n^2 - n)$ 2PS blocking pairs.
\end{proof}

\section{Top Trading Cycle (TTC)}
\label{sec:TTC}
 Top Trading Cycle (TTC) is another widely used algorithm to design strategy-proof mechanisms for finding efficient solutions for various allocation problems such as house allocation \cite{saban2013house}, school choice \cite{10.1257/000282803322157061}, course assignment~\cite{diebold2014course}, kidney exchange \cite{10.1162/0033553041382157}, etc. We define new variations of TTC for our model and analyze their properties. 
In this section, we begin by showing that a naive version of the TTC algorithm does not terminate for our model. Then, we introduce a variation called Contractual Top Trading Cycle (CTTC) that terminates, and modify it to a Contractual Top Trading Cycle with Removal (CTTCR) algorithm, producing a 4PS outcome. However, CTTCR is not Pareto optimal or strategy-proof.

\subsection{Naive TTC}
In this section, we first provide an intuitive variation of TTC for our model. 
% However, we noticed a fatal issue that this variation does not guarantee termination. We fix this problem by constructing a contractual schema, and prove that this Contractual Top Trading Cycle (CTTC) algorithm is terminating.
%
In our problem model, the preferences over both other agents and rooms need to be considered.
% there is a set of agents $N = \{a_1, a_2, \dots, a_{2n}\}$ and a set of rooms $R = \{ r_1, r_2, \dots, r_n\}$ involved,
Therefore, we redefine the way of constructing arcs in the trading graph. The algorithm (\Cref{alg:TTC1}) is split into two phases as follows:
\begin{enumerate}
    \item \textbf{Initialization}: for a given assignment $\mu$, create a directed graph $G_{\mu}$ where each agent is a vertex. Create an arc for each agent according to the following rule: an arc $(i, j)$ from agent $i$ to another agent $j$ is created when agent $i$ gets the largest utility after swapping with agent $j$. That is, $j \in \argmax_{s \in N} U_{\mu(i \leftrightarrow s)}(i)$ and $U_{\mu(i \leftrightarrow s)}(i) > U_{\mu}(i)$.
    \item \textbf{Cycle Detection and Trading}: we identify cycles in the directed graph $G_{\mu}$. A cycle $TC$ occurs when a sequence of agents $i_0, i_1, \dots, i_{k}$ exists such that agent $i_0$ has an arc to $i_1$, agent $i_1$ has an arc to $i_2$, so on, and finally, agent $i_{k}$ has an arc to $i_0$. The assignment $\mu$ is updated by {\em trading the agents according to the cycle} $TC$, i.e., agent $i_0$ is assigned to the room and roommate of agent $i_1$, and so on, finally, agent $i_{k}$ is assigned to the room and roommate of agent $i_0$. That is, we delete the triples in the cycle: $\mu = \mu \setminus \bigcup_{j=0}^{k} (i_j, \mu_N(i_j), \mu_R(i_j))$, and add the new triples:
    $\mu= \bigcup_{j=0}^k (i_j, \mu_N(i_{j+1}), \mu_R(i_{j+1})) \cup \mu$ where $j+1$ is taken mod $k+1$.  
    \item \textbf{Repeating}: repeat Phase 1 to 2 until there is no cycle in $G_{\mu}$. 
\end{enumerate}

% The notation of this intuitive TTC Variation is given in \Cref{alg:TTC1}.
\begin{algorithm}[t]
\SetAlgoNoLine
\caption{Naive Top Trading Cycle}
\label{alg:TTC1}
\KwIn{An instance $\langle N, R, H, V \rangle$ and an assignment $\mu$} 
\While{True}{
\textbf{Phase 1. Initialization}
% \BlankLine
    \\ Let $G_\mu = (V, E)$ be a directed graph where $V = N$, 
    arc $(i, j) \in E$ if and only if for $\{i,j\} \subseteq N$, $i \neq j$ and   
    $j \in \argmax_{a_s \in N} U_{\mu(i \leftrightarrow s)}(i)$  and $U_{\mu(i \leftrightarrow s)}(i) > U_{\mu}(i)$;
    \\ \textbf{Phase 2. Cycle Detection and Trading}
    % \BlankLine
    \\ \lIf{a cycle exists in $G_{\mu}$}{
        \\ \quad Let $TC=i_0, i_1, \dots, i_k$ be a cycle in $G_\mu$\;
        \quad Update $\mu$ by trading according to the cycle $TC$}
        % \State Update $\mu = (\mu \setminus \bigcup_{j=0}^{k-1} (a_{i_j}, \mu_N(a_{i_j}), \mu_R(a_{i_j})))\bigcup_{j=0}^k (a_{i_j}, \mu_N(a_{i_{j+1}}), \mu_R(a_{i_{j+1}}))$, where $j+1$ is taken mod $k$ 
    \lElse{\KwRet{The assignment $\mu$}}
    
}
\end{algorithm}

 % Our initial assumption was that this intuitive TTC variation could resolve all the trading cycles in the graph and terminate. However, 

\begin{lemma}\label{lem:CTTCNTerminate}
    Naive TTC (\Cref{alg:TTC1}) may not terminate even when the agents are indifferent about the rooms and the utilities are strict.
\end{lemma} 
The next example proves the lemma.
% \begin{table}
\begin{figure}
\begin{minipage}{0.48\linewidth}
    \centering
    % \resizebox{\linewidth}{!}{
    \begin{tabular}{|m{3mm}|m{3mm}|m{3mm}|m{3mm}|m{3mm}|m{3mm}|m{3mm}|}
    \hline
         & $a_1$ & $a_2$ & $a_3$ & $a_4$ & $r_1$ & $r_2$\\
         \hline
        $a_1$ & --{}{} & 3 & 6 & 10 & 1 & 1\\
        \hline
       $a_2$  &10 & --{}{} & 3 & 6 & 1 & 1\\
       \hline
        $a_3$ & 6 & 10 & --{}{} & 3 & 1 & 1\\
        \hline
       $a_4$ & 3 & 6 & 10 & --{}{} & 1 &  1 \\
       \hline
    \end{tabular}%}
    \captionof{table}{\small{Valuations in \Cref{eg:ttcNotTerminatingEg}}}\label{tab:ttcNotTerminatingEg}
    \end{minipage}\hfill
\begin{minipage}{0.4\linewidth}
        \centering  
        \frame{\includegraphics[width=0.8\textwidth]{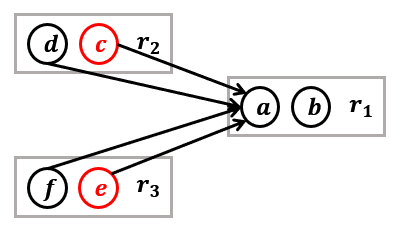}}
        \captionof{figure}{\Cref{eg:CTTCN4PS} showing CTTC is not 4PS}    
        \label{fig:CTTCn4PS}
        \vspace{5mm} 
\end{minipage}
 % \end{table}
 \end{figure}

\begin{example}\label{eg:ttcNotTerminatingEg} Consider an instance of roommate matching with four agents $a_1, a_2, a_3$, and $a_4$ and rooms $r_1$ and $r_2$. The valuations of the agents are shown in \Cref{tab:ttcNotTerminatingEg}.
    Assume that the initial assignment is $\{(a_1, a_2, r_1)$, $(a_3, a_4, r_2)\}$. We depict the new assignments constructed by our algorithm.
    \begin{itemize}\looseness-1
        \item  By trading according to the cycle $(a_1, a_3)$, we obtain assignment $\{(a_3, a_2, r_1), (a_1, a_4, r_2)\}$.
        \item In the second iteration, by trading according to  the cycle $(a_2, a_4)$, the assignment updates to $\{(a_3, a_4, r_1), (a_1, a_2, r_2)\}$ .
        \item In the third iteration, by trading according to  the cycle $(a_1, a_3)$, the assignment updates to $\{(a_1,$ $ a_4, r_1),$ $(a_2, a_3, r_2)\}$.
        \item Finally, we reach the initial assignment by trading according to the cycle $(a_2, a_4)$.
    \end{itemize}
    % In the initial assignment, the blocking pair $(a_1, a_3)$ forms a top trading cycle and the algorithm resolves it to update the assignment to $\{(a_3, a_2, r_1), (a_1, a_4, r_2)\}$. In the updated assignment, there is a blocking pair $(a_2, a_4)$, which creates another cycle and leads the continuous update to the assignment $\{(a_3, a_4, r_1), (a_1, a_2, r_2)\}$. Again, blocking pair $(a_1, a_3)$ makes the assignment being changed to $\{(a_1, a_4, r_1), (a_2, a_3, r_2)\}$. Finally, the blocking pair $(a_2, a_4)$ leads to an update to the assignment $\{(a_1, a_2, r_1), (a_3, a_4, r_2)\}$. At this point, the assignment changed back to the same as the initial assignment.
    Thus, we conclude that TTC is not terminating in this example. 

    The pattern behind this example is that for each agent, its favorite agent likes it the least. For instance, for agent $a_1$, it likes agent $a_4$ the most ($h_{a_1}(a_4) = 10$), but agent $a_4$ likes $a_1$ the least ($h_{a_4}(a_1) = 3$). When the top trading cycle assigns a top choice to an agent, the agent being assigned always has the motivation to leave the pair and form a cycle with others.
\end{example}

\subsection{First Modification of Naive TTC: CTTC}
As the naive TTC variation fails to terminate, we further modify it and call it as Contractual Top Trading Cycle (CTTC). Specifically, in the first step of \Cref{alg:TTC1}, when constructing the trading graph, we add the restrictions by applying a contractual constraint that agent's current roommate must agree with the trading. %Thus, we modify the first step of \Cref{alg:TTC1} as follows:
\begin{enumerate}
    \item \textbf{Initialization}: for a given assignment $\mu$, create a directed graph $G_\mu$ where each agent is a vertex. Create an arc for each agent according to the following rule: an arc $(i, j)$ from agent $i$ to another agent $j$ is created when agent $i$ gets the largest utility $U_{\mu(i \leftrightarrow j)}(i)$ after swapping with agent $j$ (i.e., $j \in argmax_{s \in N } U_{\mu(i \leftrightarrow s)}(i)$ and $U_{\mu(i \leftrightarrow s)}(i) > U_{\mu}(i)$.), and at the same time, agent $j$'s roommate $\mu_N(j)$ likes agent $i$ at least as much as agent $j$ (i.e., $h_{\mu_N(j)}(i) \geq h_{\mu_N(j)}(j)$).
    Phase 2 remains the same as \Cref{alg:TTC1}. The pseudocode is presented in \Cref{alg:CTTC-incomplete} in the appendix. %for completeness.
\end{enumerate}

We find that resolving a trading cycle improves social welfare, which will be used to prove that CTTC terminates. We use $ \mu_{TC}$ to denote the updated assignment upon resolving cycle $TC$ in $\mu$.
% For convenience, we use the following  notation: if there exists a top trading cycle $TC$ in an assignment $\mu$, we use $\mu_{TC}$ to denote the updated assignment after trading the agents following the trading cycle. --> moved to appendix
% We use $U_{\mu}(a)$ to denote the utility of agent $a$ in the assignment $\mu$.
\begin{restatable}{lemma}{lemSW}
\label{lem:SW}
    If there exists a top trading cycle in an assignment $\mu$ in the CTTC algorithm, then the social welfare of the assignment $\mu_{TC}$ is greater than that of $\mu$.
\end{restatable}

% \Cref{lem:SW} can be used to prove that CTTC guarantees the termination. For a given instance $\langle N,R,H,V\rangle$, there exists a maximum social welfare. As the social welfare of an assignment cannot be greater than the maximum social welfare, based on \Cref{lem:SW}, we can infer that all top trading cycles are resolved when or before reaching the maximum social welfare, and CTTC terminates when there is no trading cycle.
\Cref{lem:SW} immediately implies that CTTC is guaranteed to terminate as maximum social welfare is reached or no trading is possible. Its proof is presented in \Cref{app:CTTC}. 

Furthermore, we examine other properties of CTTC.
% including efficiency, stability and strategy-proofness. In the following 3 sections, we show that TTCR does not guarantee Pareto efficiency, may not give a 4PS solution, and is not strategy-proof, and we also provide some further tries to keep these properties from the original TTC. 

\paragraph{Pareto Optimality}
 TTC and its variants have been shown to be Pareto optimal for house allocations, kidney exchange, etc. However, \Cref{eg:ttcParetoEg} (in the appendix) show that CTTC is not Pareto optimal even under strict utilities. %, and a counterexample is given below. 
% \begin{lemma}
% The assignment output by the CTTC algorithm is not Pareto optimal.
% \end{lemma}
%
% We prove the lemma using the following example.
%
% The counterexample shows that our CTTC algorithm is not Pareto optimal in the Roommate Matching problem. Unlike the one-to-one matching problems, in our problem model, there exists a potential better assignment that cannot be achieved by one step of swapping from the current status, so the CTTC algorithm terminates without getting the optimal output.
The instance in \Cref{eg:4PSnPO} can be used to show that CTTC is not PO for binary valuations. 

\paragraph{4-Person Stability} Although TTC produces an outcome that is in the core for house allocation \cite{shapley1974cores}, a variation of TTC for school choice \cite{10.1257/000282803322157061} is not stable. Thus, we next analyze the stability of CTTC algorithm. 
% We examine the stability of CTTC (\Cref{alg:CTTC-incomplete}). 
\Cref{eg:CTTCN4PS} illustrates that the output of CTTC is not 4PS. %, and then introduce a modified CTTC called CTTCR, which gives a 4PS solution.

% \begin{lemma}
% The assignment output by  CTTC (\Cref{alg:CTTC-incomplete}) is not 4PS.
% \end{lemma}

% We prove the lemma using the following example.

\begin{example}
\label{eg:CTTCN4PS} Consider an instance with six agents $a,b,c,d,e$, and $f$, and three rooms $r_1, r_2, r_3$. Let the initial assignment be $\{(a, b, r_1) (c, d, r_2) (e, f, r_3)\}$. 
% \SR{Replace this with a table} \JL{Done} 
All agents prefer agent $b$ and room $r_1$ the most. As shown in \Cref{fig:CTTCn4PS}, there are four arcs $(c, a), (d, a), (e, a)$ and $(f, a)$, but there is no arc from agent $a$. As a result, there is no cycle and no agent is involved in a trading.  However,  a 4PS blocking pair $(c, e)$ can exist where all 4 agents $c, d, e$ and $f$ get higher utilities after trading agent $c$ with $e$. A set of valuations depicting the above is given in \Cref{tab:TTCRn4PS} in the appendix. Thus, CTTC (\Cref{alg:CTTC-incomplete}) does not output a 4PS assignment.  
\end{example}

We summarize the properties in the next corollary.
\begin{corollary}
 CTTC terminates but it does not guarantee PO or 4PS.
\end{corollary}
\subsection{A Stable Modification of CTTC: CTTCR}
\Cref{eg:CTTCN4PS} shows that the agents without outgoing arc can obstruct forming a trading cycle, even when there is a pair of agents who would prefer to exchange, and in fact, it leads to a 4PS blocking pair. 
% Next, we give a more sophisticated CTTC variation which tentatively removes the agents with no outgoing arcs. It adds back the agents to the market when we reach an assignment with a strictly higher social welfare. Furthermore, we prove that this modified CTTC terminates and produces a 4PS solution. 

% There is a sequence of cycles $\{c_{i_0}, c_{i_1},\dots,  c_{i_{k-1}}\}$ in graph $G$. 
According to \Cref{lem:SW}, when a top trading cycle is resolved and agents are traded accordingly, the social welfare of the assignment increases. Based on this fact, we take different actions according to the change in social welfare. Specifically, we remove the agents with no outgoing arc when the social welfare level remains the same and add them back when the social welfare level changes. This algorithm is based on CTTC, and we call it Contractual Top Trading Cycle with Removal (CTTCR). The details of CTTCR are described below:
    \begin{enumerate}
        % \item \label{ph:init} \textbf{Initialization}: for a given assignment $\mu$, create a directed graph $G_\mu$ where each agent is a vertex. Create an arc for each agent according to the following rule: an arc $(a_i, a_j)$ from agent $a_i$ to agent $a_j$ is created when agent $a_i$ gets the largest utility $U_{i \leftrightarrow j}(a_i)$ after swapping with agent $a_j$ (i.e., $a_j = argmax_{a_s \in N \setminus \{a_i\}} U_{i \leftrightarrow s}(a_i)$), at the same time, agent $a_j$'s roommate $\mu_N(a_j)$ likes agent $a_i$ at least as much as agent $a_j$ (i.e., $h_{\mu_N(a_j)}(a_i) \geq h_{\mu_N(a_j)}(a_j)$).
        \item \label{ph:init} \textbf{Initialization}: this step is the same as Phase 1 in CTTC. In addition, we create an empty set $V'$ for storing the removed agents. 
        
        \item\label{ph:trading_cycle} \textbf{Cycle Detection and Trading}: we identify cycles in the directed graph $G_\mu$ as in CTTC. %\Cref{alg:CTTC-incomplete}. 
        % A cycle $TC$ occurs when a sequence of agents $i_1, i_2, \dots, i_k$ exists such that agent $i_1$ has an arc to $i_2$, agent $i_2$ has an arc to $i_3$, \dots, and agent$i_k$ has an arc to $i_1$. 
        If a cycle $TC$ exists, the assignment $\mu$ is updated by trading the agents according to the cycle $TC$. The social welfare increases and we go to phase \ref{ph:SWLchanges}.
        \begin{enumerate}[label=2.\arabic*]
            \item \label{ph:SWLchanges} \textbf{Agents Adding when SW Changes}: %we add the removed agents back (if there's any) by 
            reconstruct the graph $G_\mu$ as defined in CTTC. %\Cref{alg:CTTC-incomplete}. 
        \end{enumerate}
        \item \label{ph:SWLNchanges}\textbf{Agents Removing when SW Remains the Same}: 
        % if no cycle exists, the assignment $\mu$ remains the same and social welfare does not change.
        for each agent $i$ with no outgoing arc in $G_{\mu}$, we add it 
        % and its roommate $\mu_N(i)$ 
        to $V'$ and remove it from the graph by updating vertex set $V(G)$ by setting $V(G) = N \setminus V'$. 
        % We also update the arcs by deleting all arcs incident to $V'$ as no agent can be paired with them or assigned to their room $\mu_R(i)$. 
        Finally, we construct the graph $G_{\mu}$ on $V(G)$ using the rules in the initialization step of CTTC. If all agents are removed (i.e. $V'=N$), then the algorithm returns the assignment $\mu$, otherwise, it goes to Phase \ref{ph:trading_cycle}.
    \end{enumerate}

The algorithmic description of CTTCR is given in \Cref{alg:CTTCR}.

\begin{algorithm}[t]
\SetAlgoNoLine
\caption{Contractual Top Trading Cycle with Removal}
\label{alg:CTTCR}
\KwIn{an instance $\langle N,R,H,V\rangle$ and an assignment $\mu$} 
\KwOut{a 4PS assignment} 
\BlankLine
\textbf{Phase 1. Initialization}\\
Construct $G_{\mu}$ on vertex set $V(G)=N$ as in \Cref{alg:CTTC-incomplete}\; 
Set $V'\gets \emptyset$\;

\While{$V' \neq N$}{
    % \label{ph:cycle}
    \textbf{Phase 2. Cycle Detection and Trading}\;
    \While{ there exists a cycle in $G_\mu$}{
        Let $TC = i_0, i_1, \dots, i_{k}$ be a cycle in $G_\mu$\;
        Update $\mu$ by trading according to the cycle $TC$ \;
        \textbf{Phase 2.1. Agents adding} \\
        Construct the $G_{\mu}$ on the vertex set $N$ as in  \Cref{alg:CTTC-incomplete}\;
    }
    \textbf{Phase 3. Agent Removing:}\\
    $V'\gets$ the set of vertices in $G_\mu$ with no outgoing arc\;
    Update the vertex set $V(G)$ of $G_{\mu}$:  $V(G) \gets N\setminus V'$\;
    Update the arcs of $G_\mu$: $(i, j) \in E$ if and only if $i \neq j$
    and $j \in \argmax_{s \in V(G) } U_{\mu(i \leftrightarrow s)}(i)$, $U_{\mu(i \leftrightarrow j)}(i) > U_{\mu}(i)$  and  $h_{\mu_N(j)}(i) \geq h_{\mu_N(j)}(j)$ for $\{i,j\} \subseteq V(G)$\;
}
\KwRet{Assignment $\mu$}.
\end{algorithm}

% Initialize $G_\mu$
% While ($V'\neq N$)
% While (there exists a cycle in $G_\mu$)
%     trade 
%     Initialize $G_\mu$ 
% EndWhile

% $V'$ = vertices with no outgoing arcs in $G_\mu$

% Update $G_\mu$ deleting $V'$ and reconstructing the top arcs.
% EndWhile
% Return $\mu$

% \textbf{Overview of \Cref{alg:CTTCR}}
%  \begin{enumerate}
%     \item \label{case：NinCycle}There are trading cycles but agents $a$ and $b$ are not in it. According to the rule of the algorithm, it keeps selecting and the trading process until there's no trading cycle in the graph (goes to \ref{case：noCycle}).
%     \item \label{case：noCycle} There are cycles containing only one vertex, so they are not trading cycles. According to the rule of the algorithm, all the vertices with an arc to itself are removed, and it either form a new cycle (goes to \ref{case：inCycle} or \ref{case：NinCycle}) or terminate if all the agents are removed. If the assignment process terminates, agents $a$ and $b$ have an arc to themselves, and it contradicts to $(a, c)$ is a blocking pair.
% \end{enumerate}

\noindent
\textbf{Termination and Running time}. Observe that when an assignment is updated according to a trading cycle, the social welfare increases even when $G_{\mu}$ is defined on a subset of agents. Thus, we can follow the arguments in \Cref{lem:SW} to show that the algorithm terminates.
It takes $O(n^2)$ time to construct the graph $G_\mu$.
To compute the running time of the while loop, first note that each agent has $2n-1$ possible roommates and $n$ possible room assignments. Thus, the number of distinct utilities of an agent is $O(n^2)$. Since the utility of at least one agent is strictly improved in each iteration of the while loop in Phase \ref{ph:trading_cycle}, the loop can run at most $O(n^3)$ times. Observe that an agent's utility never decreases in either Phase 2 or in Phase 3. Therefore, the outer while loop runs for at most $O(n^3)$ times. Additionally, each step in the while loop takes $O(n^2)$ time since the most expensive step is to update the graph $G_\mu$. Therefore, the final running time of the algorithm is $O(n^5)$.

% \SR{ Which example show that CTTCR \Cref{alg:CTTCR} is not PO?} \JL{Can we use the same example for CTTC? I added it below}
\paragraph{Pareto Optimality}
As CTTCR is built on top of CTTC, the same counterexample given previously can also prove that CTTCR is not PO. With more restrictions than CTTC, CTTCR also does not create any arc in the initial assignment $\mu$ in \Cref{eg:ttcParetoEg}. However, assignment $\mu'$ in \Cref{eg:ttcParetoEg} benefits all agents.

\paragraph{4-Person Stability} Although CTTCR does not guarantee PO, it satisfies 4PS. Towards this, we begin with some observations.
\begin{observation}
    \label{obs:Rule1} 
    The following properties hold for \Cref{alg:CTTCR}.
    \begin{enumerate}[label=(\alph*)]
        \item     When \Cref{alg:CTTCR} terminates, there is no trading cycle in the graph.
    \item In Phase \ref{ph:SWLNchanges}, all the agents with no outgoing arc are deleted from the graph.
    \item\label{it:noarc}  \Cref{alg:CTTCR} terminates when there is no agent with an arc.
    \end{enumerate}
\end{observation}
\begin{proof}
    The first two statements follow directly from the algorithm. For the third statement, note that the graph has no sink vertex. Thus, if there are agents that have an incident arc in the graph, then there must be a cycle in the graph.
    Therefore, when \Cref{alg:CTTCR} terminates, there is no agent with an arc.
\end{proof}

With \Cref{obs:Rule1}, we prove that CTTCR (\Cref{alg:CTTCR}) gives a 4PS solution.
\begin{theorem}
\label{the:TTCRSW_4PS}
    CTTCR (\Cref{alg:CTTCR}) outputs a $4$-person stable assignment in time $O(n^5)$.
\end{theorem}

\begin{sloppypar}
\begin{proof}
    % We prove theorem \ref{the:TTCRSW_4PS} by the contradiction that if there exists a 4PS blocking pair $(a, b)$ in the output assignment, then we show that there exists a top trading cycle, thus the algorithm would continue.
    
    % We discuss the two possibilities.
    % \begin{enumerate}
    %     \item \label{case：inCycle} Both agents $a$ and $b$ are in a trading cycle. If agents $a$ and $b$ form a shortest cycle $a \rightarrow b \rightarrow a$, then they swap and the blocking pair is resolved. 
    %     \item  Suppose agent $a$ or agent $b$ is in a cycle. We assume without loss of generality that agent $a$ is in a cycle and has an arc to another agent $c$. Then agents are traded accordingly, and $U_{a \leftrightarrow c}(a) > U_{a \leftrightarrow b}(a)$, so $(a, b)$ is no longer a 4PS blocking pair.
    %     \item Both agents $a$ and $b$ are not in a trading cycle. 
    We prove this by contradiction. Let $\mu$ be an assignment output by \Cref{alg:CTTCR}. Assume that there is a 4-person stable blocking pair $(i, j)$, and the roommates of agent $i$ and $j$ in $\mu$ are $i'$ and $j'$, respectively.
        
        We observe that agent $i$ and $j$ are in the vertex set $V(G)$ in Phase \ref{ph:SWLNchanges} of \Cref{alg:CTTCR}, as according to the definition of 4PS blocking pair we have $U_{\mu(i \leftrightarrow j)}(i) > U_{\mu}(i)$, $U_{\mu(i \leftrightarrow j)}(j) > U_{\mu}(j)$, $h_{i'}(j) > h_{i'}(i)$ and $h_{j'}(i) > h_{j'}(j)$. Thus, agent $i$ has an arc to agent $j$ or to another agent $k$ if agent $i$ gets a better utility by exchanging with agent $k$ than agent $j$ (i.e. $U_{\mu(i \leftrightarrow k)}(i) > U_{\mu(i \leftrightarrow j)}(i)$ and $h_{\mu_N(k)}(i) > h_{\mu_N(k)}(k)$). Similarly, agent $j$ also has an outgoing arc.
        % the trading graph has at least two arcs since $a$ and $b$ will have outgoing arcs to their respective top choices. 
        % Moreover, since by \Cref{obs:Rule1} the trading graph has no sink vertex, there must be a cycle in the graph. However, the algorithm terminate when there is no cycle, so the exsitence of the 4PS blocking pair contradicts the mechanism of the algorithm.
        Thus, a 4PS blocking pair contradicts the fact in \Cref{obs:Rule1}\ref{it:noarc}.
    % \end{enumerate}
\end{proof}\end{sloppypar}

% \JL{Newly added paragraph:}
\paragraph{2-Person Stability} While CTTCR achieves 4PS, we give \Cref{eg:CTTCRn2PS} (in the appendix) with values only from $\{0,1,2\}$ that shows that CTTCR is not 2PS. In fact, all $2n^2-2n$ pairs are 2PS blocking.

\paragraph{Strategy-proofness}
The original TTC is strategy-proof when applied in one-to-one assignment, where the trading cycles are disjoint because an agent can only have one outgoing arc. However, in our model, a triple consists of two agents, and a pair of agents can be involved in more than one cycle even when the utilities are strict. As a result, we need to decide which cycle to resolve first in \Cref{alg:CTTCR}. In this section, we show that agents have an incentive to misrepresent their preferences to manipulate the order of resolving cycles. Moreover, we give an example to prove the impossibility of designing a cycle selection rule that gives strategy-proofness.

\begin{remark}
    Let $\mathcal{F}$ be a rule to select a cycle for trading in Phase \ref{ph:trading_cycle} of \Cref{alg:CTTCR}. We observe from \Cref{eg:CTTCnSP} that there exists no rule $\mathcal{F}$ such that \Cref{alg:CTTCR} is strategy-proof.
\end{remark}

\begin{example}
\label{eg:CTTCnSP}
    Consider an instance with ten agents $a_1$ to $a_{10}$, and five rooms $r_1$ to $r_5$. Agents' valuations are given in  \Cref{tab:TTCRSD_nSP-agent,tab:TTCRSD_nSP-room}. As shown in  \Cref{fig:TTCSD_nSP1}, the initial assignment is $\{(a_1, a_2, r_1)$, $(a_3, a_4, r_2)$, $(a_5, a_6, r_3)$, $(a_7, a_8, r_4)$, $(a_9, a_{10}, r_5)\}$. There are two top trading cycles, and the algorithm needs to select one of them to resolve first. Agent $a_1$ (or $a_2$) can prevent forming the cycle $c_1$ (resp. $c_2$) by misreporting its preferences over agent $a_3$ (resp. $a_8$). We can check that when one of the cycles is selected and agents are traded accordingly, then the other cycle no longer exists after the trade. \Cref{fig:TTCSD_nSP3} shows that agent $a_1$ gets a higher utility ($h_{a_1}(a_7)+v_{a_1}(r_5)=22$) than what it gets by truly reporting its values ($h_{a_1}(a_3)+v_{a_1}(r_2)=20$), and  \Cref{fig:TTCSD_nSP2} shows that agent $a_2$ gets a higher utility ($h_{a_2}(a_4)+v_{a_2}(r_3)=22$) than what it gets by truly reporting its values ($h_{a_2}(a_8)+v_{a_2}(r_4)=20$). As a result, no matter which cycle is selected, there exists an agent who has an incentive to lie. Thus, CTTCR is not strategy-proof with any cycle selection rule.
\end{example}

% \begin{figure}
%     \centering
%     \includegraphics[width=0.7\linewidth]{imgs/eg_TTCSD_nSP.PNG}
%     \caption{Caption}
%     \label{fig:eg_TTCRSD_nSP}
% \end{figure}

\begin{figure}[t]
    \centering
    \begin{subfigure}{0.31\textwidth}
        \centering
        \frame{\includegraphics[width=\textwidth]{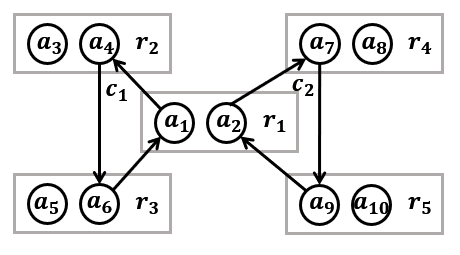}}
        \caption{The Initial Assignment}
        \label{fig:TTCSD_nSP1}
        \vspace{5mm}
    \end{subfigure}
    \hfill
    \begin{subfigure}{0.31\textwidth}
        \centering
        \frame{\includegraphics[width=\textwidth]{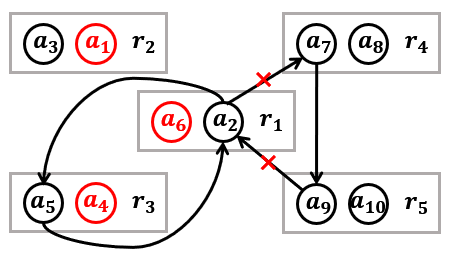}}
        \caption{Cycle $c_1$ is selected.}
        \label{fig:TTCSD_nSP2}
        \vspace{5mm}
    \end{subfigure}
    \hfill
    \begin{subfigure}{0.31\textwidth}
        \centering
        \frame{\includegraphics[width=\textwidth]{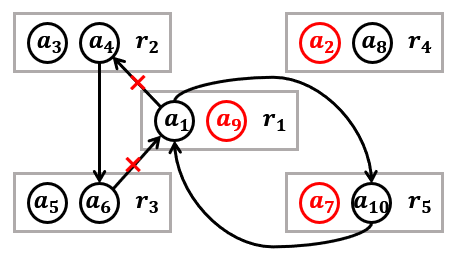}}
        \caption{Cycle $c_2$ is selected.}
        \label{fig:TTCSD_nSP3}
        \vspace{5mm}
    \end{subfigure}
    \caption{\Cref{eg:CTTCnSP} showing CTTCR is not strategy-proof}
    \label{fig:TTCSD_nSP}
        \vspace{5mm}
\end{figure}

\begin{table}[t]
\centering
 % \resizebox{.99\linewidth}{!}{
 \begin{minipage}{0.48\linewidth}
\centering
\begin{tabular}{|m{3mm}|m{3mm}|m{3mm}|m{3mm}|m{3mm}|m{3mm}|m{3mm}|m{3mm}|m{3mm}|m{3mm}|m{3mm}|}
\hline
    &$a_1$&$a_2$&$a_3$&$a_4$&$a_5$&$a_6$&$a_7$&$a_8$&$a_9$&$a_{10}$\\
     \hline
    $a_1$ &--{}{} & 1 & 13 &1&1&1& 16 &1&12&1\\
    \hline
    $a_2$ &1 &--{}{}&2&16&3&12&4&13&5&6\\
   \hline
    $a_3$ &10&1&--{}{}&9&1& 1 & 1 &1& 1&1\\
    \hline
   $a_4$ &1&14&1&--{}{}&13&1&1&1&1&1\\
   \hline
   $a_5$ &2&1&1&10&--{}{} &9&1&1&1&1\\
   \hline
   $a_6$ &1& 13& 1&1& 14 &--{}{}&1& 1 & 1 &1\\
   \hline
   $a_7$ &14&1&1&1&1&1&--{}{}&1&1&13\\
   \hline
   $a_8$ &1&10&1&1&1&1&9&--{}{}&1&1\\
   \hline
   $a_9$ &13&1&1&1&1&1&1&1&--{}{}&14\\
   \hline
   $a_{10}$ &1&1&1&1&1&1&10&1&9&--{}{}\\
   \hline
\end{tabular}
 % }
\caption{Agent values in \Cref{eg:CTTCnSP}} %: a counterexample of CTTCR is not strategy-proof}
\label{tab:TTCRSD_nSP}\label{tab:TTCRSD_nSP-agent}
 \end{minipage}
 \hfill
 \begin{minipage}{0.48\linewidth}
 \centering
 % \resizebox{.99\linewidth}{!}{
\begin{tabular}{|m{3mm}|m{3mm}|m{3mm}|m{3mm}|m{3mm}|m{3mm}|}
\hline
    &$r_1$&$r_2$&$r_3$&$r_4$&$r_5$\\
     \hline
    $a_1$ &9&7&2&1&6\\
    \hline
    $a_2$ &9&2&6&7&1\\
   \hline
    $a_3$ &1&7&1&1&1 \\
    \hline
   $a_4$ &1&1&7&1&1\\
   \hline
   $a_5$ &9&1&7&1&1\\
   \hline
   $a_6$ &7&1&1&1&1\\
   \hline
   $a_7$ &1&1&1&1&7\\
   \hline
   $a_8$ &1&1&1&7&1\\
   \hline
   $a_9$ &7&1&1&1&1\\
   \hline
   $a_{10}$ &9&1&1&1&7\\
   \hline
\end{tabular}
 % }
\caption{Room values in \Cref{eg:CTTCnSP}} %: a counterexample of CTTCR is not strategy-proof}
\label{tab:TTCRSD_nSP-room}    
 \end{minipage}
\end{table}

% From the SD algorithm it may seem that when a 4PS assignment is reached, the assignment have maximum welfare or at least the assignment would be Pareto optimal. However, in the definition of 4PS, it is required that the pair agents who are accepting the new swapped roommates must strictly improve along with the pair of agents who are swapping. Thus, 4PS does not implicitly guarantee PO as we will show in \Cref{ex:noPObinaryswap}.
\vspace{-3mm}
\section{Binary \& Symmetric Valuations}
\label{sec:sym_bin}
% \SR{ADD Motivation for binary and Symmetric preference}\SR{added in our contribution}

In this section, we finally design an algorithm that produces a 2PS assignment. We will show that, in contrast to general preferences, a 2PS assignment can be found in polynomial time by a \textit{swapping} algorithm when the preferences are binary and symmetric. The definition of the swap operation is the same as the one given in \Cref{sec:model}.

A challenge of the algorithm we design is to determine the pair to be swapped first among the 2PS blocking pairs, as an agent can be involved in more than one blocking pair. In \Cref{alg:Swapping}, we denote a \emph{selection rule} as $\mathcal{F}$ that selects the pair of agents to be swapped. For instance, it can arbitrarily select one blocking pair to swap, or assign the agents with a certain order and select the lexicographically smallest pair to swap.
% \SR{write that we will consider strategy-proofness and selection rules for that} \JL{Done} 
We examine the strategy-proofness of the algorithm when applying different selection rules.

We first prove the following result for binary and symmetric valuations. Its complete proof and \Cref{ex:non-binswap} showing that it does not hold under non-binary valuations are presented in \Cref{app:symmbin}.

%For this variation model, we have 2 results and an algorithm for 2-person stable solution.

\begin{restatable}{theorem}{theSBSW}
\label{the:SB_SW}
    If valuations are symmetric and binary, then swapping two agents who formed a 2PS blocking pair increases the social welfare of the assignment.
\end{restatable}
\begin{proof}[Proof Sketch]
Suppose $\mu$ is an arbitrary assignment, and triples $t_1=(a, b, r_1)$ and $t_2=(c, d, r_2)$ are in the assignment $\mu$, where $a,b,c,d \in N$ and $r_1,r_2 \in R$. We assume that there exists a 2PS blocking pair $(a, c)$. 
 We aim to show that the assignment $\mu'$, created from $\mu$ 
 by resolving the blocking pair $(a,c)$ by \emph{swapping agent $a$ and $c$}, has a greater social welfare than $\mu$.  
 Observe that the new triples in $\mu'$ are $t_1' = (b, c, r_1)$ and $t_2' = (a, d, r_2)$.  Let $SW$ and $SW'$ denote the social welfare of $\mu$ and $\mu'$, respectively.
 % , without the triples changed during swap.

Next, we will define the utility of agents $a$ and $c$ and compare the social welfare of $\mu$ and $\mu'$.
The existence of a 2PS blocking pair $(a,c)$ implies 
$h_a(d) + v_a(r_2) > h_a(b) + v_a(r_1)$ from utility of agent $a$ and $h_c(b) + v_c(r_1) > h_c(d) + v_c(r_2)$ from utility of agent $c$.
% equations \ref{eq:u_a} and \ref{eq:u_c}. 
% \begin{equation}
% \label{eq:u_a}
%     \text{From utility of }a: h_a(d) + v_a(r_2) > h_a(b) + v_a(r_1)
% \end{equation}
% \begin{equation}
% \label{eq:u_c}
%     \text{From utility of }c: h_c(b) + v_c(r_1) > h_c(d) + v_c(r_2) 
% \end{equation}
% For simplification, let $t_1$ and $t_2$ be the two triples involving blocking pair, and we exclude them from the assignment $\mu$, 
%
We represent their difference by $X$. Thus, $X = (h_a(d) + v_a(r_2) + h_c(b) + v_c(r_1)) - (h_a(b) + v_a(r_1) + h_c(d) + v_c(r_2)) > 0$.
% \begin{multline*}
% X = (h_a(d) + v_a(r_2) + h_c(b) + v_c(r_1)) \\- (h_a(b) + v_a(r_1) + h_c(d) + v_c(r_2)) > 0
% \end{multline*}

Recall that social welfare is the sum of the utilities of all the agents. 
Subtracting the social welfare of $\mu$ from that of $\mu'$, we can show
$SW' - SW = (h_c(b) - h_c(d) + h_a(d) - h_a(b)) + X$.
We show that in the binary case, $X \geq 2$, so $SW' - SW \geq 0$.  Thus, we prove the theorem.
\end{proof}

\paragraph{Swapping Algorithm for 2PS}
Starting from an arbitrary assignment, we swap a 2PS blocking pair that is selected using a selection rule $\mathcal{F}$, until we find a 2PS assignment. Theorem \ref{the:SB_SW} indicates that the process of swapping agents who form a 2PS blocking pair terminates in time $O(n)$ as the maximum social welfare for a given binary instance $\langle N, R, H, V\rangle$ with $n$ rooms is at most $4n$. 
% We use a selection rule $\mathcal{F}$ to select the 2PS blocking pair for swapping. 

% \begin{algorithm}
% \begin{algorithmic}
% \caption{Swapping Algorithm (Symmetric and Binary)}
% \label{alg:Swapping}
%     \Require A roommate matching instance $\langle N,R,H,V\rangle$, where $\forall i, j \in N: h_i(j) = h_j(i) = 1$ or $h_i(j) = h_j(i) = 0$, $\forall k \in N, r \in R: v_k(r) = \{0, 1\}$ and an assignment $\mu$
%     \Ensure A $2$PS roommate matching
%     \State Let $G = (V, E)$ be an undirected graph where $V = N \bigcup R$ and 
%     edge $(i, j) \in E$ if $h_i(j) = h_j(i) = 1$, edge$(i, r) \in E$ if $v_i(r) = 1$.
%     \While{$\exists \;i, j \in N$ s.t. $U_{\mu(i \leftrightarrow j)}(i) > U_{\mu}(i)$ and $U_{\mu(i \leftrightarrow j)}(j) > U_{\mu}(j)$}
%         \State Applying \textit{Selection Rule} $F$ to select one of them.
%         \State Update $\mu = (\mu \setminus 
%         ((i, \mu_N(i), \mu_R(i)), (j, \mu_N(j), \mu_R(i)))$
%         \State $\bigcup ((i, \mu_N(j), \mu_R(j)), (j, \mu_N(i), \mu_R(i)))$. 
%     \EndWhile
%     \\ \Return{A set of triples $\mu$ as an assignment}.
% \end{algorithmic}
% \end{algorithm}

 \begin{remark}
    The hardness reduction from \citet{chan2016assignment} showed that finding a solution that is 2PS is NP-hard when the preferences are non-symmetric but the agents are indifferent between the rooms. Our \Cref{alg:Swapping} shows that the hardness is mainly due to the asymmetry in the preferences as, otherwise, the problem can be solved in polynomial time using  \Cref{the:SB_SW}.
\end{remark}

A detailed description is given in \Cref{alg:Swapping} in the appendix. However, we provide a simple counterexample in the appendix to show that it is not PO.
When valuations are symmetric, agents can only lie about their preference for rooms. We further consider the strategy-proofness of the algorithm, and found that the order of the swapping matters and the \textit{Selection Rule} $\mathcal{F}$ will affect the strategy-proofness as shown in \Cref{ex:eg_SB_n_SP_2}. 
We get the following corollary from the above example.
% However, we give \Cref{con:nSP} and leave it an open question if there exists any selection rule to make \Cref{alg:Swapping} strategy-proof.
\vspace{-2mm}
\begin{corollary}
    When we use serial dictatorship as the \textit{Selection Rule $\mathcal{F}$}, the Swapping algorithm is not strategy-proof.
 \end{corollary}

\section{Conclusion}
In this paper, we explore the problem of stability in roommate matching with externalities. We consider both 2-person stability and 4-person stability while also analyzing efficiency and strategy-proofness. 
We modify two classical algorithms, serial dictatorship and top trading cycle to fit our model. Our result demonstrates that while the SD-based algorithm ensures strategy-proofness, PO, and 4PS, it fails to be 2PS. Meanwhile, 
although our TTC-based algorithm achieves 4PS, it is not PO, and we show its impossibility of being strategy-proof, which underscores a key limitation of TTC in simultaneously guaranteeing stability and incentive compatibility in our problem model. 
% We also identify a symmetric and binary preference structure where 2PS assignments are tractable. However, it remains open to design a selection rule to make our swap based algorithm strategy-proof in this setting.  

Our findings highlight the inherent trade-offs between stability, efficiency, and strategic behavior in roommate matching problem, suggesting a number of future directions. For example, an important direction is to explore the weaker notions of strategy-proofness that can be potentially achieved. Additionally, it would be valuable to design an FPT algorithm for 2PS. Next, it is natural to define a \emph{weak 4PS blocking} where only the swapping agents need to strictly improve and the roommates are not worse off. 
% Then, our SD algorithm is not 4PS but CTTCR is. 
Finally, applying our framework to practical roommate and housing assignment platforms could provide insights into real-world strategic behavior and potential adjustments to existing mechanisms. 

\section*{Acknowledgment}
We thank the anonymous reviewers for their useful comments.

\bibliography{arxiv}

\newpage
%Appendix
\appendix
% \twocolumn[\section*{\hspace{0.4\linewidth}Appendix\\ \hspace{0.4\linewidth}\small Paper \#7091}]
\section{An algorithm for 4PS Assignment}
\label{app:DS_LS}
\citet{chan2016assignment} developed an algorithm to produce a 4PS assignment that has social welfare at  least $2/3$ of the optimal. It uses first an algorithm called  Double Matching (DM) and then the Local Searching (LS) algorithm running in time $O(n^3)$  and $O(n^5)$, respectively. 
Given an instance $\langle N, R, H, V \rangle$, the Double Matching algorithm constructs two maximum weight perfect matchings. First one (denoted by $M_1$) is between the agents in the graph $G_1$ on agents $N$ such that for two agents $i, j \in N$, $i \neq j$, there is an edge $(i, j)$ with the weight $h_{i}(j) + h_{j}(i)$. The second one (denoted by $M_2$) is a 1-2-bipartite matching between the agents and rooms in the graph $G_2$ with the vertex set $N \cup R$ and for each agent $i \in N$ and each room $r \in R$, there is an edge $(i,r)$ with the weight $v_{i}(r)$. Finally, to combine the two matchings, they consider each cycle (of length more than two) formed by the edges $M_1 \cup M_2$ of the two matchings and take the edges that correspond to a pairwise disjoint set of triples (each triple has two agents and a room) of highest total weight.
% Given an instance $\langle N, R, H, V \rangle$, define an undirected weighted graph $G$ with the vertex set $N \cup R$ and edge set specified as follows. For two agents $i, j \in N$, $i \neq j$, there is an edge $(i, j)$ with the weight $h_{ij} + h_{ji}$. For each agent $i \in N$ and each room $r \in R$, there is an edge $(i,r)$ with the weight $v_{ir}$. Denote $E_H = \{(i,j) \; |
% \; i, j \in N, i \neq j\}$ and $E_V = \{(i,r) \; | \; i \in N,r \in R\}$. Then finding a maximum social welfare assignment is equivalent to finding a maximum-weight triangle partition of $G$ such that each triangle comprises one edge in $E_H$ and two edges in $E_V$.
The complete algorithm is given in \Cref{alg:DM}.
Given the outcome of Double Matching as input, their Local Search algorithm repeatedly resolves the 4PS blocking pairs by swapping the agents involved in the blocking, giving 4-person stable outcome.

Although Chan et al. \cite{chan2016assignment} designed a greedy algorithm following the serial dictatorship mechanism to find a 4PS solution efficiently, we provide a counterexample to prove that it is not strategy-proof.
% \begin{proposition}
%     Double Matching and Local Searching algorithms in \cite{chan2016assignment} are not strategy-proof.
% \end{proposition}
\propDM*
We use the following example to prove the statement.
\begin{example}\label{eg:doublematching}
We give an example of one cycle from the set $M_1 \cup M_2$, as shown in  \Cref{fig:DM1}. 
% For simplification, we give the values of different edges in  \Cref{tab:edges} instead of agents' preference values. 
A set of partial preference values being used are given in \Cref{tab:DM_n_SP} where the values in the empty cells could be assumed to be any number smaller than $1$.
The edge weights are computed using the valuations. For instance, the weight of edge $e_1 = (a_1,r_1)$ is %between agent $a_1$ and room $r_1$ is agent $a_1$'s valuation
$v_{a_1}(r_1) = 4$, and the weight of edge $e_2=(a_1,a_2)$ %between two agents $a_1$ and $a_2$
 equals to %the sum of their preference values for each other, that is,
$w(e_2)= h_{a_1}(a_2) + h_{a_2}(a_1) = 2 + 1$. According to the Double Matching algorithm, we get the following weights of the triples in the cycle shown in \Cref{fig:DM1}:
$W_0 = w(e_3) + w(e_6) + w(e_9) = 3+2+2 = 7$; 
$W_1 = w(e_1) + w(e_4) + w(e_7) = 4+1+1 = 6$; $W_2 = w(e_2) + w(e_5) + w(e_8) = 3+4+2 = 9$.

Therefore, $\min_{t\in \{0,1,2\}} W_t = W_1$, and we remove the edges $e_1, e_4$ and $e_7$, outputting an assignment $\{(a_3, a_4, r_1)$, $(a_1, a_2, r_2)$, $(a_5, a_6, r_3)\}$, as shown in \Cref{fig:DM2}. %, where edges $e_1', e_4'$ and $e_7'$ are newly generated edges. 

However, suppose agent $a_1$ lies about its valuation towards room $r_1$. Let agent $a_1$ misreports the valuation to be $v'_{a_1}$ such that $v'_{a_1}(r_1) = 6$, as shown by the red line in figure \ref{fig:DM3}. All other values remain the same. Under the misreported valuation, the new weight $w(e_1) = 6$ and $W_1 = 8$. As a result, $W_0$ has the minimum weight as opposed to $W_1$ and the algorithm will remove edges in $W_0$.
%rather than the edges in $W_1$.
Thus, the assignment will be $\{(a_1, a_2, r_1)$, $(a_5, a_6, r_2)$, $(a_3, a_4, r_3)\}$ where $a_1$ gets strictly better utility. This example shows that the agents have the incentives to lie to get their favorite rooms or roommates when applying DM, so DM is not strategy-proof. As the Local Searching algorithm uses DM as its foundation to find a 4PS assignment, their algorithm for 4PS is also not strategy-proof.
\end{example}

\begin{figure}
    \centering
    \begin{subfigure}{0.31\textwidth}
        \centering
        \includegraphics[width=\textwidth]{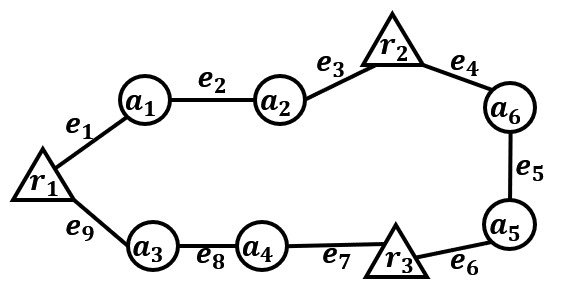}
        \caption{The Original Cycle in $M_1 \cup M_2$}
        \label{fig:DM1}
        \vspace{5mm}
    \end{subfigure}
    \hfill
    \begin{subfigure}{0.31\textwidth}
        \centering
        \includegraphics[width=\textwidth]{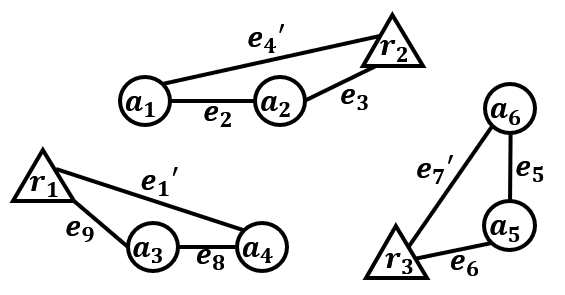}
        \caption{Applying Double-Matching Algorithm}
        \label{fig:DM2}
        \vspace{5mm}
    \end{subfigure}
    \begin{subfigure}{0.31\textwidth}
        \centering
        \includegraphics[width=\textwidth]{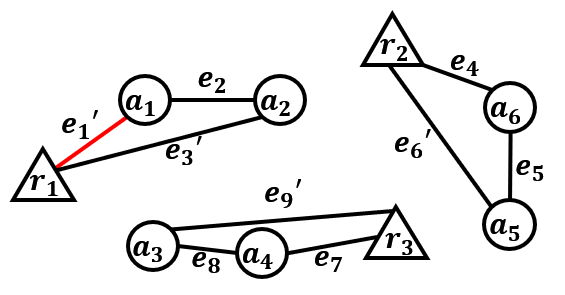}
        \caption{A case of $a_1$ lying}
        \label{fig:DM3}
        \vspace{5mm}
    \end{subfigure}
    \caption{\small{An Example Cycle and the Application of Double-Matching Algorithm}}
    \label{fig:DMEg}
        \vspace{5mm}
\end{figure}

% \begin{table}[b]
%     \centering
%     \begin{tabular}{|ccccccccc|}
%     \hline
%         $e_1$ & $e_2$ & $e_3$ & $e_4$ & $e_5$ & $e_6$ & $e_7$ & $e_8$ & $e_9$ \\
%         \hline
%          4 & 3 & 3 & 1 & 4 & 2 & 1 & 2 & 2 \\
%     \hline
%     \end{tabular}
%     \caption{Edge Values}
%     \label{tab:edges}
% \end{table}

\begin{table}[b]
\small
    \centering
    \begin{tabular}{|m{4mm}|m{3mm}|m{3mm}|m{3mm}|m{3mm}|m{3mm}|m{3mm}|m{5.9mm}|m{4mm}|m{3mm}|}
    \hline
         & $a_1$ & $a_2$ & $a_3$ & $a_4$ & $a_5$ & $a_6$ & $r_1$ & $r_2$ &$r_3$\\
         \hline
        $a_1$ &--{}{} & 2 &  & & &  & 4 \textcolor{red}{(6)} &  &\\
        \hline
        $a_2$  &1 &--{}{}&   &   &   &   &   & 3 & \\
       \hline
        $a_3$ &  &  &--{}{}& 1 &  && 2 & &\\
        \hline
       $a_4$ &  &  & 1 &--{}{}&  &   &  & &1\\
       \hline
       $a_5$ &  &  &  & &--{}{} & 2 &  &  &2\\
       \hline
       $a_6$ &  &  &  &  & 2  &--{}{}&  & 1& \\
       \hline
    \end{tabular}
    \caption{Valuations: an example of DM not being strategy-proof.}
    \label{tab:DM_n_SP}
\end{table}

% \section{Stable and Strategy-proof Mechanism}
% The lack of the strategy-proofness of the previous algorithm inspires us to study two algorithms -- Serial Dictatorship (SD) and Top Trading Cycle (TTC). These two algorithms have been extensively used to design strategy-proof mechanisms for finding stable, and/or efficient solutions for various allocations problems such as house allocation \cite{saban2013house,abdulkadirouglu1998random}, school choice \cite{10.1257/000282803322157061}, course assignment~\cite{diebold2014course}, kidney exchange\cite{10.1162/0033553041382157}. In the following sections, we develop algorithms by modifying SD and TTC for our model, and examine their properties.

% \section{Algorithms from \citet{chan2016assignment}}
\begin{algorithm}
\SetAlgoNoLine
\caption{Double-Matching (DM) \cite{chan2016assignment}}\label{alg:DM}
Find a maximum-weight perfect matching $M_1$ of $G_1 = (N,E_H)$\;
Find a maximum-weight perfect 1-2-matching $M_2$ in bipartite graph $G_2 =(N,R,E_V)$ \;
Let $C = \{C_1,C_2...,C_k\}$ the set of cycles of $M_1 \cup M_2$ \;
\For{$i = 0$ to $k$}{ 
    Denote the edges of $C_i$ by $e_1,...,e_{3l}$ in a cyclic order \;
    (starting from an arbitrary vertex)\;
    \If{$l\geq 2$}{ 
        \For{$t=0$ to 2}{
            $W_t = \sum_{j=1,..., 3l: \, j\equiv t \, mod \; 3}w(e_j)$ \;
        }
        Let $t^\ast \in \{0,1,2\}$ be a minimizer of $min_t W_t$\;
        Remove edges $e_j$ for all $j \equiv t^\ast$ mod 3 \;
    }
}
The subgraph is now a collection of vertex-disjoint connected triples $(i, j, r)$ \;
\KwRet{The set of these triples $(i,j,r)$ as an assignment}
\end{algorithm}
\begin{algorithm}
\SetAlgoNoLine
\caption{Local Search (LS) \cite{chan2016assignment}}\label{alg:LS}
Start from assignment outputted by Double-Matching mechanism\;
\While{there is a 4PS blocking pair $(i,j)$}{
    Swap the people $i$ and $j$\;
}
\KwRet{The current room assignment}
\end{algorithm}

\section{Explanatory Example of SD and Omitted Proofs}

\begin{algorithm}
\SetAlgoNoLine
\KwIn{A roommate matching instance $\langle N,R,H,V\rangle$} 
\KwOut{A 4PS assignment} 
\BlankLine
Let $\sigma = (a_1, a_2, \dots, a_{2n})$ denote an fixed sequence of agents\;
Set assignment $\mu = \emptyset$\;
Let $N' = \{a_1, a_2, \dots, a_{2n}\}$ be the set of unmatched agents in $\mu$\;
Let $R' = \{r_1, r_2, \dots, r_n\}$ be the set of unmatched rooms in $\mu$\;
    \For{ $i=1$ to $2n$}{
        \If{$a_i$ in $N'$}{
            % \State Let $a^\ast = \argmax_{a_j \in N'} h_{a_i}(a_j)$ \Comment{i.e., $h_{a_i}(a^\ast) >  h_{a_i}(a_j)$, $\forall a_j \in N' \setminus \{a^\ast\}$.}
            Let $a^\ast$ be a maximizer with $h_{a_i}(a^\ast) \geq  h_{a_i}(a_j)$, $\forall a_j \in N'$\;
            Let $r^\ast$ be a maximizer with $v_{a_i}(r^\ast) \geq  v_{a_i}(r_k)$,  $ \forall r_k \in R'$\;
            Add the triple $(a_i, a^\ast, r^\ast)$ to $\mu$\; 
            Delete $a_i$ and $a^\ast$ from $N'$\;
            Delete $r^\ast$ from $R'$\;
        }
    }
\KwRet{A set of triples as an assignment $\mu$}
\caption{Serial Dictatorship}\label{alg:SD4PS}
\end{algorithm}

% \textbf{Overview of \Cref{alg:SD4PS}}.
The mechanism of the Serial Dictatorship algorithm is very straightforward. In our model,the number of agents is twice the number of rooms. The "dictator" in each round has to choose its preferred room and another agent as its roommate, and the final allocation assigns two agents to each room.  We calculate the running time of \Cref{alg:SD4PS}. The loop runs $2n$ times, and each iteration takes $O(n)$-time to find the maximizer $a^\ast$ and $r^\ast$. Therefore, the total time complexity for the SD algorithm is $O(n^2)$. 
Additionally, we present \Cref{eg:SD} to show an instance of applying \Cref{alg:SD4PS}.  
The serial dictatorship algorithm has several desirable properties that we  consider in the following. 

\paragraph{Pareto Optimality}
% \SR{Write lemma statement \Cref{alg:SD4PS} is Pareto efficient.DONE}
We show that the outcome of \Cref{alg:SD4PS} is Pareto optimal for strict utility profiles, that is, the utility of an agent over the pairs of room and roommate forms a strict ordering. 

\begin{lemma}
\label{lem:SD_PE}
    Any assignment output by \Cref{alg:SD4PS} is Pareto optimal for strict utility values. 
\end{lemma}
\begin{proof}
    We show that there is no other assignment that can make some agents better off without making others worse off. Suppose that 
    there is a triple $(j, j', r_j)$ in the output assignment $\mu$, and
    there exists an assignment $\mu'$ where agent $j$ is the first agent in the ordering $\sigma$ who   strictly improves and no other agent is worse off. Suppose that $j$ is assigned to a more preferred room $r_i$ (or a preferred roommate $i'$) in $\mu'$.  As the mechanism of SD is that agent $j$ chooses its most preferred room $r_j$ (or agent $j'$) from the remains of the market, so agent $j$ can only be better off by choosing a room or agent that was already chosen and removed by a previous round. Therefore, in the sequence $\sigma$, we have $i \succ j$ (i.e., agent $i$ has a higher priority to make the selection than agent $j$) such that in the output assignment $\mu$, room $r_i$ (resp. agent $i'$) is assigned to agent $i$.  
 However, since $j$ is the first agent in $\sigma$ who strictly improves in $\mu$, agent $i$ does not strictly improve. Since we assumed that the utility values are strict and agent $i$ chose the best room and roommate available in the SD algorithm, agent $i$ can only worsen when room $r_i$ (or agent $i'$) is assigned to agent $j$. Thus, the output assignment is Pareto optimal.
 
 However, the PO poperty does not hold when the utility values are not strict, because it is possible for agent $i$ to keep the same utility in assignment $\mu'$ when $h_i(i') = h_i(j')$ or $v_i(r_i) = v_i(r_j)$. 
 
\end{proof} 

\paragraph{Strategy-proofness}
Another important property of the original SD algorithm is the strategy-proofness. We prove that the SD algorithm keeps this property when being extended for our model. 
\begin{lemma}\label{lem:SD4PS-SP}
    \Cref{alg:SD4PS} is strategy-proof.
\end{lemma}
\begin{proof}[Proof of \Cref{lem:SD4PS-SP}]
\begin{sloppypar}
    Without loss of generality, we assume that there is an agent $j$ such that $i \succ j$ in the order $\sigma$, and agent $j$ is matched with agent $p$ and room $q$ in the output assignment $\mu$ if it honestly reports its preferences. 
    
    Assume that agent $j$ misreports its preferences, leading to an assignment $\mu'$, and it claims that it prefers agent $p'$ over agent $p$ (or it prefers room $q'$ over room $q$). When it is agent $j$'s turn to be the ``dictator'', the remaining agent set is $N' \subset N$ and the remaining room set is $R' \subset R$. We consider two cases separately. 
    \begin{enumerate}
        \item Suppose agent $p'$ (or room $q'$) was selected by agent $i$. Then agent $p'$ (resp. room $q'$) is not in the set $N'$ (resp. $R'$). According to the mechanism, agent $j$ cannot select it, so it goes to the next top choices, and agent $p$ and room $q$ are assigned to agent $j$. The output assignment $\mu'$ is the same as assignment $\mu$. Thus, $j$ does not improve by misreporting.
        \item Agent $p'$ (or room $q'$) is still in the set $N'$ (or $R'$). According to the mechanism, agent $p'$ and room $q'$ are assigned to agent $j$. But the utility of agent $j$ in assignment $\mu'$ is smaller than its utility in assignment $\mu$, i.e. $h_{j}(p') < h_{j}(p)$ (or $v_{j}(q') < v_{j}(q)$).
    \end{enumerate}
    As a result, agent $j$ get either an equal or a smaller utility by lying. So agents have no incentive to lie and \Cref{alg:SD4PS} is strategy-proof.
\end{sloppypar}
\end{proof}

\begin{example}
\label{eg:SD}
As shown in figure \ref{fig:SDeg1}, the agent set $N = \{a, b, c, d, e, f\}$, and the room set $R = \{i, j, k\}$. We simply determine a sequence that is the same as the alphabetical order, and the agents' preference values over other agents and rooms are given in table \ref{tab:SDeg}. The first round is shown in figure \ref{fig:SDeg2}, agent $a$ is the dictator and it selects its most preferred roommate $c$ and room $i$. Room $i$, agent $a$ and $c$ are removed from the market. Then it comes to the second round as shown in figure \ref{fig:SDeg3}, agent $b$ is the dictator and it picks the most preferred agent $f$ and room $j$ from the remaining. In the final round, agent $d$ selects agent $e$ and room $k$, and the final assignment is shown in figure \ref{fig:SDeg4}.
\end{example}
\begin{figure}
        \centering
        \begin{subfigure}[b]{0.23\textwidth}
            \centering
            \frame{\includegraphics[width=\textwidth]{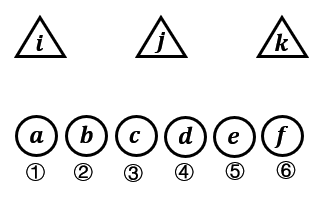}}
            \caption{{\small Initial market}}    
            \label{fig:SDeg1}
        \vspace{5mm}
        \end{subfigure}
        \begin{subfigure}[b]{0.23\textwidth}  
            \centering 
            \frame{\includegraphics[width=\textwidth]{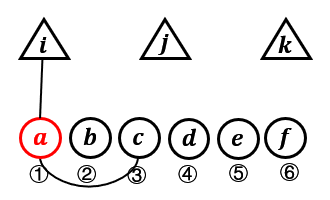}}
            \caption{{\small Round 1}}    
            \label{fig:SDeg2}
        \vspace{5mm}
        \end{subfigure}
        \begin{subfigure}[b]{0.23\textwidth}   
            \centering 
            \frame{\includegraphics[width=\textwidth]{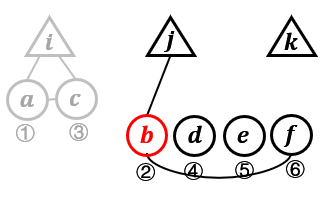}}
            \caption{{\small Round 2}}    
            \label{fig:SDeg3}
        \vspace{5mm}
        \end{subfigure}
        \begin{subfigure}[b]{0.23\textwidth}   
            \centering 
            \frame{\includegraphics[width=\textwidth]{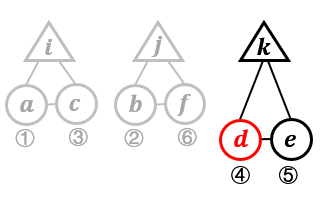}}
            \caption{{\small Final allocation}}    
            \label{fig:SDeg4}
        \vspace{5mm}
        \end{subfigure}
        \caption{\small An explanatory example of SD} 
        \label{fig:SDeg}
        \vspace{5mm}
    \end{figure}

\begin{table}
    \centering
    \begin{minipage}{0.48\textwidth}
    \centering
    \begin{tabular}{|c|c|c|c|c|c|c|c|c|c|}
    \hline
         & \textbf{a} & \textbf{b} & \textbf{c} & \textbf{d} & \textbf{e} & \textbf{f} & \textbf{i} & \textbf{j} & \textbf{k} \\
         \hline
        \textbf{a} & --{}{} & 5 & 7 & 4 & 4 & 2 & 5 & 3 & 1\\
        \hline
        \textbf{b} & 7 & --{}{} & 5 & 1 & 2 & 3 & 3 & 4 & 3\\
        \hline
        \textbf{c} & 2 & 2 & --{}{} & 4 & 3 & 1 & 2 & 5 & 3\\
        \hline
        \textbf{d} & 4 & 6 & 2 & --{}{} & 1 & 3 & 3 & 4 & 4\\
        \hline
        \textbf{e} & 3 & 5 & 2 & 2 & --{}{} & 6 & 3 & 1 & 4\\
        \hline
        \textbf{f} & 7 & 2 & 4 & 5 & 6 & --{}{} & 4 & 2 & 2\\
        \hline
    \end{tabular}
    \caption{\small{Explanatory example of SD: preference values}}
    \label{tab:SDeg}
    \end{minipage}
    \hfill
    \begin{minipage}{0.48\textwidth}
    \centering
    \begin{tabular}
    {|m{3mm}|m{3mm}|m{3mm}|m{3mm}|m{3mm}|m{3mm}|m{3mm}|m{3mm}|m{3mm}|m{3mm}|}
    \hline
         & $a_1$ & $a_2$ & $a_3$ & $a_4$ & $a_5$ & $a_6$ & $r_1$ & $r_2$ &$r_3$\\
         \hline
        $a_1$ &--{}{} & 5 & 4 &3 &2 & 1 & 1 & 1&1\\
        \hline
        $a_2$  &1 &--{}{}& 5 & 4 & 3 & 2 & 1 & 1&1\\
       \hline
        $a_3$ & 5 & 4 &--{}{}& 3 & 2 & 1& 1 & 1&1\\
        \hline
       $a_4$ & 5 & 4 & 1 &--{}{}& 3 &  2 & 1 & 1&1\\
       \hline
       $a_5$ & 5 & 4 & 3 & 2&--{}{} &  1 & 1 & 1&1\\
       \hline
       $a_6$ & 5 & 4 & 3 & 2 &  1 &--{}{}& 1 & 1&1\\
       \hline
    \end{tabular}
    \caption{Vauations in \Cref{eg:SD2PSblocking}: a worst-case scenario example using SD}
    \label{tab:SD2PSnumber}
    \end{minipage}
\end{table}

\subsection{SD is not 2PS}
Serial dictatorship offers a low level of equality because it disadvantages agents with low priority, overlooking their preferences. This unequal treatment of the agents results in instability, as low-priority agents may be assigned to a roommate or room they dislike, leading to dissatisfaction and a lack of adherence to the assignment.
Additionally, we give an example of the worst-case scenario illustrated by \Cref{thm:SD2PSblocking}. 
\begin{example}\label{eg:SD2PSblocking}
Let the priority order $\sigma = \{a_1, a_2, a_3, a_4, a_5, a_6 \}$, where the size of agent set $2n=6$. The preference values are given in \Cref{tab:SD2PSnumber}. For simplification, we let all the agents be indifferent about the room (i.e., give the same preference value 1 to all rooms). The assignment output by the serial dictatorship algorithm is $\{(a_1, a_2, r_1)$, $(a_3, a_4, r_2)$, $(a_5, a_6, r_3)\}$. 

The output assignment has blocking pairs: $(a_2, a_3)$, $(a_2, a_4)$, $(a_2, a_5)$, $(a_2, a_6)$, $(a_4, a_5)$, and $(a_4, a_6)$, and the number of the blocking pairs is exactly equal to $3^2 - 3 = 6$.
\end{example}

% \begin{table}[t]
%     \centering
%     \begin{tabular}{|m{3mm}|m{3mm}|m{3mm}|m{3mm}|m{3mm}|m{3mm}|m{3mm}|m{3mm}|m{3mm}|m{3mm}|}
%     \hline
%          & $a_1$ & $a_2$ & $a_3$ & $a_4$ & $a_5$ & $a_6$ & $r_1$ & $r_2$ &$r_3$\\
%          \hline
%         $a_1$ &--{}{} & 5 & 4 &3 &2 & 1 & 1 & 1&1\\
%         \hline
%         $a_2$  &1 &--{}{}& 5 & 4 & 3 & 2 & 1 & 1&1\\
%        \hline
%         $a_3$ & 5 & 4 &--{}{}& 3 & 2 & 1& 1 & 1&1\\
%         \hline
%        $a_4$ & 5 & 4 & 1 &--{}{}& 3 &  2 & 1 & 1&1\\
%        \hline
%        $a_5$ & 5 & 4 & 3 & 2&--{}{} &  1 & 1 & 1&1\\
%        \hline
%        $a_6$ & 5 & 4 & 3 & 2 &  1 &--{}{}& 1 & 1&1\\
%        \hline
%     \end{tabular}
%     \caption{Vauations in \Cref{eg:SD2PSblocking}: a worst-case scenario example using SD}
%     \label{tab:SD2PSnumber}
%         \vspace{5mm}
% \end{table}

\section{Modification of TTC}
\subsection{The CTTC algorithm}\label{app:CTTC}
We present the algorithm in detail here.
Therefore, we get \Cref{alg:CTTC-incomplete}.
\begin{algorithm}
\SetAlgoNoLine
\caption{Contractual Top Trading Cycle}
\label{alg:CTTC-incomplete}
\KwIn{an instance $\langle N, R, H, V \rangle$ and an assignment $\mu$} 
\While{True}{
    \textbf{Phase 1. Initialization}
    \\Let $G_\mu = (V, E)$ be a directed graph where $V = N$,
    arc $(i, j) \in E$ if and only if $i \neq j$ and 
    $j \in \argmax_{s \in N} U_{\mu(i \leftrightarrow s)}(i)$, $U_{\mu(i \leftrightarrow j)}(i)> U_{\mu}(i)$, and 
    $h_{\mu_N(j)}(i) \geq h_{\mu_N(j)}(j)$ for $\{i,j\} \subseteq N$;
    \\ \textbf{Phase 2} is the same as \Cref{alg:TTC1}\;
    % \\ \Return{The set of triples $\mu$ as an assignment}
    }
\end{algorithm}

\paragraph{Proof \cref{lem:SW}}
For convenience, we use the following notation: if there exists a top trading cycle $TC$ in an assignment $\mu$, we use $\mu_{TC}$ to denote the updated assignment after trading the agents following the trading cycle.
\lemSW*
\begin{proof}[Proof of \cref{lem:SW}]
\begin{sloppypar}
    Suppose there exists a top trading cycle $TC = (i_1, i_2, \dots, i_{k})$. Recall that the assignment $\mu_{TC} = (\mu \setminus \bigcup_{j=1}^k (i_j, \mu_N(i_j), \mu_R(i_j))) \cup \bigcup_{j=1}^k (i_j, \mu_N(i_{j+1}), \mu_R(i_{j+1}))$, where $j+1$ is taken mod $k$. There are three categories of agents, namely, agents who are in the trading cycle, their partners, and others. We examine their utilities respectively. 
    
      Let $b$ be an agent in the trading cycle, and there is an arc $(b, c)$. We claim that $U_{\mu_{TC}}(b) > U_{\mu}(b)$. According to the mechanism of CTTC, an arc $(b, c)$ is created only if $c = \argmax_{s \in N} U_{\mu(b \leftrightarrow s)}(b)$. As a result, the claim holds because $U_{\mu_{TC}}(b) = U_{\mu(b \leftrightarrow c)}(b) > U_{\mu}(b)$.

    Let $b$ be a partner of an agent who is in the trading cycle, then we claim that $U_{\mu_{TC}}(b) \geq U_{\mu}(b)$. Let $\mu_N(i_{j+1})=b$, where $i_{j+1} \in TC$. According to the mechanism of the CTTC, this claim holds because an arc $(i_j, i_{j+1})$ is created only if $b$ accepts this arc. Therefore, $b$ prefers $i_j$ at least as much as $i_{j+1}$, i.e., $h_{b}(i_j) \geq h_{b}(i_{j+1})$. The rooms agent $b$ assigned to are the same in $\mu$ and in $\mu_{TC}$. As a result, we have the result that $U_{\mu_{TC}}(b) \geq U_{\mu}(b)$.
    
    Let $b$ be an agent s.t. $b \notin TC$ and $b \neq \mu_N(i_j)$ s.t. $i_j \in TC$. Since the roommate and room assigned to $b$ are the same in $\mu$ and in $\mu_{TC}$, the utilities of $b$ are the same in $\mu$ and in $\mu_{TC}$: $U_{\mu}(b) = U_{\mu_{TC}}(b)$.

    Thus, after resolving a trading cycle and the assignment is updated from $\mu$ to $\mu_{TC}$, the utilities of agents in the trading cycles increase, the utilities of their roommates increase or stay the same, and the utilities of other agents stay the same. As a result, the sum of utilities (i.e. the social welfare) of $\mu_{TC}$ is greater than that of $\mu$. 
\end{sloppypar}
\end{proof}

\paragraph{Pareto Optimality}
\begin{table}[t]
\centering
        % \begin{minipage}{0.48\linewidth}
        \centering
        % \resizebox{\linewidth}{!}{
        \begin{tabular}{|m{3mm}|m{3mm}|m{3mm}|m{3mm}|m{3mm}|m{3mm}|m{3mm}|}
    \hline
         & $a_1$ &$a_2$  & $a_3$ & $a_4$ & $r_1$ & $r_2$\\
         \hline
        $a_1$ & --{}{} & 7 & 1 &  2& 3 & 5\\
         \hline
       $a_2$  & 7 & --{}{} & 2 & 1 & 3 & 5\\
         \hline
        $a_3$ & 1 & 2 & --{}{} & 7 & 5 & 3\\
         \hline
       $a_4$  & 2 & 1 &  7& --{}{} & 5 & 3\\
         \hline
    \end{tabular}%}
    \caption{\small{Valuations in \Cref{eg:ttcParetoEg}}}
    \label{tab:ttcParetoEg}
    % \end{minipage}
 \end{table}
\begin{example}\label{eg:ttcParetoEg}
We consider 4 agents and 2 rooms in this example, and the preference values are given in table \ref{tab:ttcParetoEg}. Let the initial assignment $\mu = \{(a_1, a_2, r_1) (a_3, a_4, r_2)\}$. One can check that all the agents get a utility of 10 (7 from roommate and 3 from room) and prefer their current allocation in the sense that there is no arc in the trading graph (as defined in \Cref{alg:CTTC-incomplete}).
% positions over other possible positions in the assignment.
%
As a result, there is no trading cycle and the CTTC outputs the initial assignment directly. However, there exists an assignment $\mu' = \{(a_3, a_4, r_1)$, $(a_1, a_2, r_2)\}$ where each agent improves (every agent gets a utility of $12$). Thus, the initial assignment is not PO.
\end{example}

\paragraph{4-Person Stability} CTTC is not 4PS.
\begin{table}[h]
\centering
\begin{tabular}{|m{3mm}|m{3mm}|m{3mm}|m{3mm}|m{3mm}|m{3mm}|m{3mm}|m{3mm}|m{3mm}|m{3mm}|}
\hline
     & $a$ & $b$ & $c$ & $d$ & $e$ & $f$ & $r_1$ & $r_2$ & $r_3$\\
     \hline
    $a$ &--{}{} & 9 & 5 & 4 & 3 & 2 & 3 & 2 & 1\\
    \hline
    $b$ &9 &--{}{}& 3  & 5  & 7  &4  & 3  & 2  & 1 \\
   \hline
    $c$ &  2 &  9 &--{}{}& 1 & 5 & 7  & 3 & 1  &  2 \\
    \hline
   $d$ & 5 & 9  & 1  &--{}{}& 7  &  3 & 3 & 1  & 2\\
   \hline
   $e$ & 5  &  9 & 4 &  7 &--{}{} &  1 &  3 & 2 & 1 \\
   \hline
   $f$ &  4 &9  &  7 & 2 & 1 &--{}{}& 3 & 2  &  1 \\
   \hline
\end{tabular}%}
\caption{Valuations in \Cref{eg:CTTCN4PS}}
\label{tab:TTCRn4PS}
        \vspace{5mm}
% \end{minipage}
\end{table}

\subsection{The CTTCR Algorithm}

\paragraph{Intuitive Idea of \Cref{alg:CTTCR}}.
The algorithm is based on TTC. 
We modify the original TTC by creating a contractual schema where the roommates are also happy with the trade. With this restriction, resolving a trading cycle leads to the increase in social welfare. Since there is a maximum social welfare for a given assignment, the algorithm terminates. 
However, the new restriction means that some agents may be left with no outgoing arc, and some 4PS blocking pairs may be left unresolved when the algorithm terminates. To avoid this, we delete the agents with no arc and ensure the existence of a trading cycle in the graph. 
Moreover, after resolving a trading cycle, the deleted agents can potentially form a 4PS blocking pair with the agents left in the graph. Thus, we add them back and update the graph to consider all the possibilities. The level of social welfare is a good indicator here. We trade the agents according to a top trading cycle. Agents with no arc are deleted from the graph when social welfare level keeps unchanged (i.e. no trading cycle in the graph), and deleted agents are added back to the graph when social welfare changes. 

\paragraph{2-Person Stability}
% \begin{table}[t]
%     \centering
%     \begin{tabular}{|m{3mm}|m{3mm}|m{3mm}|m{3mm}|m{3mm}|m{3mm}|m{3mm}|}
%     \hline
%          & $a$ & $b$ & $c$ & $d$ & $r_1$ & $r_2$\\
%          \hline
%         $a$ &--{}{} & 3 & 1 &5 &2 & 3 \\
%         \hline
%         $b$  &5 &--{}{}& 3 & 1 & 3 & 2\\
%        \hline
%         $c$ & 1 & 5 &--{}{}& 3& 3 & 2 \\
%         \hline
%        $d$ & 3 & 1 & 5 &--{}{}& 2 & 3  \\
%        \hline
%     \end{tabular}
%     \caption{Valuations in \Cref{eg:CTTCRn2PS}}
% \label{tab:CTTCRn2PS}
% \end{table}
 \begin{example}\label{eg:CTTCRn2PS}
Consider an instance with $2n$ agents $a_1,a_2,\dots, a_{2n}$ and $n$ rooms $r_1,r_2, \dots, r_n$. For each $i \in [2n]$, if $i$ is odd, then agent $a_i$ values agent $a_{i+1}$ at $1$ and every other agent at zero; it values room $r_{\frac{i+1}{2}}$ at zero and every other room at $2$. For each $i \in [2n]$, if $i$ is even, then agent $a_i$ values agent $a_{i-1}$ at $1$ and every other agent at zero; it values room $r_{\frac{i}{2}}$ at zero and every other room at $2$. Suppose the initial assignment  $\mu=\{(a_1,a_2,r_1), (a_3,a_4,r_2), \dots, (a_{2n-1}, a_{2n}, r_{n})\}$.

Then no agent is part of a 4PS blocking as for each agent $a_i$ they get utility $1$ in the current assignment $\mu$ but if its roommate swaps with another agent, it would get utility $0$.  Consequently, there is no outgoing arc from any of the agents in $G_{\mu}$ because of the contractual nature of the arcs in $G_{\mu}$ (as described in the Initialization phase of \Cref{alg:CTTC-incomplete}). Thus, the CTTCR algorithm terminates. However, each agent is in a 2PS blocking with every other agent that is not its roommate. For an agent $a_i$ who swaps with agent $a_j$, the utility of both $a_i$ and $a_j$ increases from one to two as they value the other person's room at $2$. Thus, we have that each agent is part of $(2n-2)$ 2PS blocking pairs. Therefore, in total, there are $2n^2-2n$ 2PS blocking pairs in $\mu$.
    % Consider an instance of roommate matching with four agents $a, b, c,$ and $d$, and two rooms $r_1$ and $r_2$. The preference values are given in \Cref{tab:CTTCRn2PS}.
    % Assume that the initial assignment $\mu$ is $\{(a, b, r_1), (c, d, r_2)\}$. There is a 2PS blocking pair $(a, c)$ as $h_a(d) + v_a(r_2) > h_a(b) + v_a(r_1)$ and $h_c(b) + v_c(r_1) > h_c(d) + v_c(r_2)$. However, when we apply \Cref{alg:CTTCR} and construct the trading  graph $G_{\mu}$, there's no outgoing arc from agent $a$ (or $c$) as its roommate agent $b$ (or $d$) is worse off after the swap. Besides, there is no arc from agent $b$ or $d$ as they get the favorite rooms and roommates in the initial assignment. As a result, there is no cycle in the trading graph, and the algorithm goes to Phase \ref{ph:SWLNchanges}. All the agents $a, b, c$ and $d$ have no arc and are removed from the graph $G_{\mu}$. The algorithm terminates and outputs the initial assignment. The existence of the 2PS blocking pair $(a, c)$ shows that the algorithm fails to be 2PS.
\end{example}

% \paragraph{Strategy-proofness}
% \Cref{tab:TTCRSD_nSP-agent} and \Cref{tab:TTCRSD_nSP-room} show the valuations in \Cref{eg:CTTCnSP} which shows that CTTCR is not strategy-proof.

\section{Symmetric and Binary Preference}\label{app:symmbin}
\paragraph{Proof of \Cref{the:SB_SW}}
\theSBSW*
\begin{proof}[Proof of \Cref{the:SB_SW}]
Suppose $\mu$ be an arbitrary assignment, and triples $t_1=(a, b, r_1)$ and $t_2=(c, d, r_2)$ are in the assignment $\mu$, where $a,b,c,d \in N$ and $r_1,r_2 \in R$. We assume that there exists a 2PS blocking pair $(a, c)$. 
 We aim to show that the assignment $\mu'$, created from $\mu$ 
 by resolving the blocking pair $(a,c)$ by \emph{swapping agent $a$ and $c$}, has a greater social welfare than $\mu$.  
 Observe that the new triples in $\mu'$ are $t_1' = (b, c, r_1)$ and $t_2' = (a, d, r_2)$.  Let $SW$ and $SW'$ denote the social welfare of $\mu$ and $\mu'$, respectively. 
 % without the triples changed during swap.

Next, we will define the utility of agents $a$ and $c$ and compare the social welfare of $\mu$ and $\mu'$.
The existence of a 2PS blocking pair $(a,c)$ implies equations \ref{eq:u_a} and \ref{eq:u_c}. 
\begin{equation}\label{eq:u_a}
% \nonumber %
    \text{From utility of }a: h_a(d) + v_a(r_2) > h_a(b) + v_a(r_1) %\tag{1}
\end{equation}
\begin{equation}\label{eq:u_c}
% \nonumber %
    \text{From utility of }c: h_c(b) + v_c(r_1) > h_c(d) + v_c(r_2) %\tag{2}
\end{equation}
Recall that social welfare is the sum of utilities from all the agents. 
% For simplification, let $t_1$ and $t_2$ be the two triples involving blocking pair, and we exclude them from the assignment $\mu$, 
We use $SW{(\mu \setminus \{t_1, t_2\})}$ to represent the sum of utilities of all the agents except for agents $a, b, c$ and $d$. The social welfare of assignment $\mu$ and $\mu'$ are as follows:
\begin{multline}
\label{eq:SW}
    SW = SW{(\mu \setminus \{t_1, t_2\})} + h_a(b) + h_b(a) + v_a(r_1) + v_b(r_1)\\ + h_c(d) + h_d(c) + v_c(r_2) + v_d(r_2)
\end{multline}
\begin{multline}
\label{eq:SW'}
    SW' = SW{(\mu' \setminus \{t_1', t_2'\})} + h_b(c) + h_c(b) + v_b(r_1) + v_c(r_1)\\ + h_a(d) + h_d(a) + v_a(r_2) + v_d(r_2)
\end{multline}
With the property of symmetric, we transform equations \ref{eq:SW} and \ref{eq:SW'} into the following equations:
% \begin{equation}
\begin{dmath}
    SW = SW{(\mu \setminus \{t_1, t_2\})} + 2h_a(b) + 2h_c(d) + v_a(r_1) + v_b(r_1) + v_c(r_2) + v_d(r_2)    
% \end{equation}
\end{dmath}
% \begin{equation}
\begin{dmath}
    SW' = SW{(\mu' \setminus \{t_1', t_2'\})} + 2h_c(b) + 2h_a(d) + v_b(r_1) + v_c(r_1) + v_a(r_2) + v_d(r_2)
\end{dmath}
% \end{equation}
The swap changes the triples $t_1, t_2$ in $\mu$ to the triples $t_1', t_2'$ in $\mu'$, and there are only 4 agents involved in these triples.  As a result, all the agents except for agents $a, b, c$ and $d$ have the same utilities in $\mu'$ as their utilities in $\mu$, indicating that $SW{(\mu \setminus t_1, t_2)} = SW{(\mu' \setminus t_1', t_2')}$. 

Equations \ref{eq:u_a} and \ref{eq:u_c} can be transformed to the following and we represent their difference by $X$:
\begin{multline*}
X = (h_a(d) + v_a(r_2) + h_c(b) + v_c(r_1)) - (h_a(b) + v_a(r_1) + h_c(d) + v_c(r_2)) > 0
\end{multline*}
Subtracting the old social welfare from the new social welfare, the result can be split into two parts. We have known that $X$ is greater than zero, and we discuss the remaining in different cases. 
\begin{dmath}
\label{eq:SW'_2}
    SW' - SW = 2h_c(b) + 2h_a(d) + v_c(r_1) + v_a(r_2) -
    [2h_a(b) + 2h_c(d) + v_a(r_1) + v_c(r_2)] 
    = (h_c(b) - h_c(d) + h_a(d) - h_a(b)) + X
\end{dmath}

 An agent $i$ has a valuation $v_i(j)$ for another agent $j$ and  $h_i(j) \in \{0, 1\}$. As a result, the value of $h_c(b) - h_c(d) \in \{-1, 0, 1\}$ (similar for $h_a(d) - h_a(b)$). In binary case, equations \ref{eq:u_a} and equation \ref{eq:u_c} indicate that $X \geq 2$, so $SW' - SW \geq 0$.  Thus, we prove the theorem.
\end{proof}

\paragraph{Non-binary valuations}
\begin{proposition}
    \Cref{the:SB_SW} does  not hold for non-binary preferences.
\end{proposition}
We prove the proposition using the counterexample given in \Cref{ex:non-binswap}.
\begin{remark}
For \textit{Non-binary preferences},  we know that $X$ is greater than $0$ in the proof of \Cref{the:SB_SW}. However, it does not guarantee that $SW' - SW$ is greater than $0$. The term $SW' - SW$ is non-negative when the preferences are 
    (i) symmetric and binary; or,
    (ii) symmetric and the agents are indifferent about the rooms.  We provide more details in the counterexample given in \Cref{ex:non-binswap}.
\end{remark}

\begin{example}\label{ex:non-binswap}
We provide a counterexample showing that when valuation between agents is symmetric but non-binary, the social welfare of the assignment can decrease after swapping two agents who formed a 2PS blocking pair. 
%The original matching is $\mu: (a, b, r_1) (c, d, r_2)$ and there is a blocking pair $(a, c)$, leading to a new matching $\mu': (c, b, r_1) (a, d, r_2)$ after swapping agent $a$ and $c$. 

\begin{table}[t]
    \centering
    \begin{tabular}{|m{3mm}|m{3mm}|m{3mm}|m{3mm}|m{3mm}|m{3mm}|m{3mm}|}
    \hline
         & $a$ & $b$ & $c$ & $d$ & $r_1$ & $r_2$\\
         \hline
        $a$ &--{}{} & 4 & 1 &2 &1 & 4 \\
        \hline
        $b$  &4 &--{}{}& 2 & 1 & 4 & 1\\
       \hline
        $c$ & 1 & 2 &--{}{}& 4& 4 & 1 \\
        \hline
       $d$ & 2 & 1 & 4 &--{}{}& 1 & 4  \\
       \hline
    \end{tabular}
    \caption{Preference values: an example when valuations are symmetric and non-binary}
\label{tab:eg_nonBi}
\end{table}

The initial assignment $\mu = \{(a, b, r_1), (c, d, r_2)\}$, and the preference values are given in table \ref{tab:eg_nonBi}. For agent $a, h_a(d) + v_a(r_2) = 2 + 4 \; > \; h_a(b) + v_a(r_1) = 4 + 1$.
For agent $c, h_c(b) + v_c(r_1) = 2 + 4 \; > \; h_c(d) + v_c(r_2) = 4 + 1$.
We get a new assignment $\mu'$ after swapping the blocking pair $(a,c)$, and we can calculate the difference between social welfare $SW' - SW$ by equation \ref{eq:SW'_2}.
First we calculate $X$ as follows, $X = (h_a(d) + v_a(r_2) + h_c(b) + v_c(r_1)) - (h_a(b) + v_a(r_1) + h_c(d) + v_c(r_2))$. Thus, using the preference values we get that $X = (2 + 4 + 2 + 4) - (4 + 1 + 4 + 1)=2$. Finally, $ SW' - SW 
    = (h_b(c) - h_b(a)) + (h_d(a) - h_d(c)) + X$. Therefore, $ SW' - SW=-2$.
% \begin{equation}
% \begin{split}
%     X & = (v_a(d) + \hat{v}_a(r_2) + v_c(b) + \hat{v}_c(r_1)) - (v_a(b) + \hat{v}_a(r_1) + v_c(d) + \hat{v}_c(r_2)) \\
%     & = (2 + 4 + 2 + 4) - (4 + 1 + 4 + 1)\\
%     & = 2
% \end{split}
% \end{equation}
% \begin{dmath}
%     X = (h_a(d) + v_a(r_2) + h_c(b) + v_c(r_1)) - (h_a(b) + v_a(r_1) + h_c(d) + v_c(r_2)) = (2 + 4 + 2 + 4) - (4 + 1 + 4 + 1)= 2
% \end{dmath}

% \begin{dmath}
%     SW' - SW 
%     = (h_b(c) - h_b(a)) + (h_d(a) - h_d(c)) + X 
%     = (2-4) + (2-4) + 2 
%     = - 2
% \end{dmath}

When the value of $h_b(a)$ is much larger than $h_b(c)$ (similarly $h_d(c) > h_d(a)$), the result is negative. 
It cannot happen in the binary case because $h_b(a) - h_b(c) \leq 1$, $h_d(c) - h_d(a) \leq 1$, but $X \geq 2$.
Alternatively, when agents are indifferent about rooms, where $v_a(r_1) = v_a(r_2)$, $v_c(r_1) = v_c(r_2)$, $X = h_a(d) - h_a(b) + h_c(b) - h_c(d) > 0$. Thus, $SW - SW' = 2X > 0$.
\end{example}

\subsection{Algorithm for binary and symmetric valuations}
\textbf{Overview of \Cref{alg:Swapping}}.
Intuitively, the algorithm finds a blocking pair and swaps the agents in the blocking pair to construct a new assignment.
The termination can only happen when there is no 2PS blocking pair, and this indicates that the algorithm \ref{alg:Swapping} gives a 2PS solution. 
% The valuations are binary, so they are not strict (i.e. an agent can like other agents and rooms equally). When shown in the graph, 

\begin{algorithm}
\SetAlgoNoLine
    \KwIn{A roommate matching instance $\langle N,R,H,V\rangle$, where $\forall i, j \in N: h_i(j) = h_j(i) = 1$ or $h_i(j) = h_j(i) = 0$, $\forall k \in N, r \in R: v_k(r) = \{0, 1\}$ and an assignment $\mu$}
    \KwOut{A $2$PS assignment}
    \BlankLine
    Let $G = (V, E)$ be an undirected graph where $V = N \bigcup R$ and 
    edge $(i, j) \in E$ if $h_i(j) = h_j(i) = 1$, edge$(i, r) \in E$ if $v_i(r) = 1$\;
    \While{$\exists \;i, j \in N$ s.t. $U_{\mu(i \leftrightarrow j)}
    (i) > U_{\mu}(i)$ and $U_{\mu(i \leftrightarrow j)}(j) > U_{\mu}(j)$}{
        Applying \textit{Selection Rule} $\mathcal{F}$ to select one of them\;
        Update $\mu = (\mu \setminus 
        ((i, \mu_N(i), \mu_R(i)), (j, \mu_N(j), \mu_R(i)))$
        $\bigcup ((i, \mu_N(j), \mu_R(j)), (j, \mu_N(i), \mu_R(i)))$\;
    }
    \KwRet{A set of triples as an assignment $\mu$}
\caption{Swapping Algorithm (Symmetric and Binary)}
\label{alg:Swapping}
\end{algorithm}

\paragraph{Strategy-proofness}
We provide an example that involves more than one agent who has the incentive to lie.
\begin{example}\label{ex:eg_SB_n_SP_2}
\begin{figure*}
        \centering
        \begin{subfigure}[b]{0.98\textwidth}
            \centering
            \includegraphics[width=\textwidth]{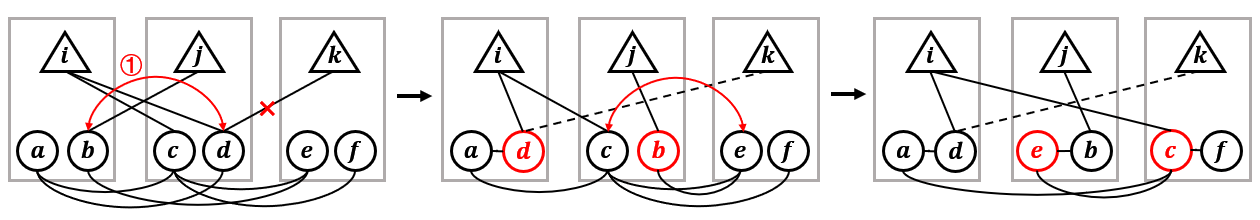}
            
            \caption{{\small Agent $d$ lies.}}    
            \label{fig:eg_SB_n_SP_2_1}
        \vspace{5mm}
        \end{subfigure}
        \begin{subfigure}[b]{0.98\textwidth}  
            \centering 
            \includegraphics[width=\textwidth]{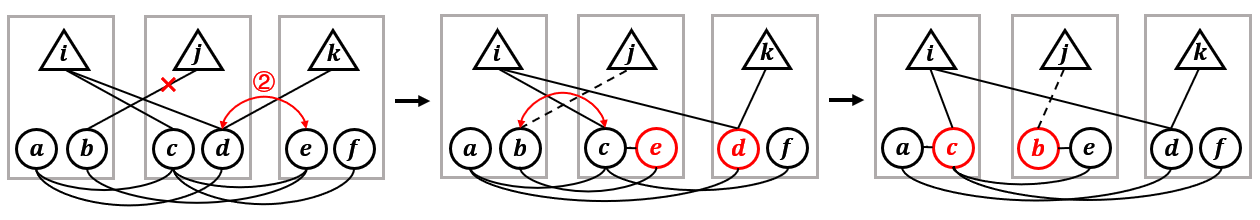}
            \caption{{\small Agent $b$ lies.}}    
            \label{fig:eg_SB_n_SP_2_2}
        \vspace{5mm}
        \end{subfigure}
        
        \caption{\small Example: more than 1 agent has incentive to lie in Symmetric \& Binary case.} 
        \label{fig:eg_SB_n_SP_2}
        \vspace{5mm}
    \end{figure*}
     There are agents $a, b, c, d, e, f \in N$, and room $i, j, k \in R$. The initial assignment is $\{(a, b, i)$, $(c, d, j)$, $(e, f, k)\}$. There are multiples blocking pairs $(a, c)$, $(a, d)$, $(b, c)$, $(b, d)$, $(d, e)$ and $(d, f)$ in the initial assignments. A \textit{Selection Rule $\mathcal{F}$} is applied to decide the pair to be swapped first.

    % We discuss some of the blocking pairs in the initial assignment.  
    We discuss the agent's incentive to lie based on the blocking pair selected by the selection rule.
    Agent $d$ has the incentive to lie about its preference over room $k$ to prevent the forming of pair $2$ in figure \ref{fig:eg_SB_n_SP_2_2}. Because the swap of pair $2$ would increase the utility of agent $d$ by $1$, however, other order of swapping (e.g. figure \ref{fig:eg_SB_n_SP_2_1}) can increase its utility by $2$. In this case, an agent has the incentive to lie when involved in more than one blocking pair, where one of them gives it both a favorite room and a favorite roommate, and another only gives it a favorite room.
    In addition, agent $b$ has the incentive to misreport its preference over room $j$ to prevent the form of pair $1$ in figure \ref{fig:eg_SB_n_SP_2_1}. Because other order of swapping (e.g. figure \ref{fig:eg_SB_n_SP_2_2}) can give it a higher utility. In this case, an agent sees that a swap can potentially assign its favorite roommate (i.e. agent $e$) to its favorite room (i.e. room $j$), so it prevents the swap involving itself (i.e. pair 1) from happening first by lying. 

    Among all blocking pairs, we found that agent $b$ has an incentive to lie when \textit{Rule $\mathcal{F}$} selects $(b, d)$ or $(b, c)$ to swap first. Agent $d$ has incentive to lie when \textit{Rule $\mathcal{F}$} selects $(a, d)$, $(d, e)$, or $(d, f)$ to swap first. But there is no agent who has an incentive to lie if \textit{Rule $\mathcal{F}$} selects $(a, c)$ to swap first.
\end{example}

We conjecture that the above corollary can be extended for any selection rule $F$.
\begin{conjecture}
\label{con:nSP}
    \Cref{alg:Swapping} is not strategy-proof for any \textit{Selection Rule $\mathcal{F}$}.
\end{conjecture}

\paragraph{Pareto Optimality} The following example shows that \Cref{alg:Swapping} is not Pareto optimal.
\begin{example}\label{ex:noPObinaryswap}
    For agents $a, b, c, d \in N$ and room $i, j \in R$, $h_b(c) = h_c(b) = 1$ and all other valuations not given are $0$. The initial assignment $\mu$ is $\{(a, b, i)$ $(c, d, j)\}$. According to the rules of \Cref{alg:Swapping}, no swap happens and the output assignment is still $\mu$. Every agent in assignment $\mu$ has the utility of 0. However, an assignment $\mu' = \{(b, c, i), (a, d, j)\}$ can give agent $b$ and agent $c$ better utilities of 1 without worsening any other agents. As a result, \Cref{alg:Swapping} is not Pareto optimal. 
\end{example}

\end{document}